\documentclass{article}
\usepackage[utf8]{inputenc}
\pdfoutput = 1
\usepackage{fullpage,amsthm,amsmath,natbib,algorithm,algorithmic,enumitem,afterpage,amssymb,setspace,amsfonts,array,verbatim,commath,newfloat}
\usepackage{graphicx}
\usepackage[dvipsnames]{xcolor}
\usepackage[hang,flushmargin]{footmisc}
\usepackage[colorlinks=TRUE,citecolor=Blue,linkcolor=BrickRed,urlcolor=PineGreen]{hyperref}
\newtheorem{lemma}{Lemma}
\newtheorem{theorem}{Theorem}

\allowdisplaybreaks

\newcommand{\bs}[1]{\boldsymbol{#1}}
\def\bhat{\hat{\beta}}
\def\shat{\hat{s}}
\def\vbhat{\bs{\bhat}}
\def\vshat{\bs{\shat}}
\def\vbeta{\bs{\beta}}
\def\ash{{ASH}}
\DeclareMathOperator*{\diag}{diag}
\DeclareMathOperator*{\var}{var}
\DeclareMathOperator*{\argmax}{arg\ max}
\DeclareMathOperator*{\argmin}{arg\ min}
\DeclareMathOperator*{\tr}{tr}
\newcommand*{\doi}[1]{\href{http://dx.doi.org/#1}{doi: #1}}

\begin{document}

\singlespacing
\title{Empirical Bayes Shrinkage and False
  Discovery Rate Estimation, Allowing For Unwanted Variation}

\author{David Gerard$^{1}$ and Matthew Stephens$^{1,2}$\\
\small Departments of Human Genetics$^1$ and Statistics$^2$, \\
\small  University of Chicago, Chicago, IL, USA}

\date{}
\maketitle

\begin{abstract}
  We combine two important ideas in the analysis of large-scale
  genomics experiments (e.g.~experiments that aim to identify genes
  that are differentially expressed between two conditions). The first
  is use of Empirical Bayes (EB) methods to handle the large number of
  potentially-sparse effects, and estimate false discovery rates and
  related quantities. The second is use of factor analysis methods to
  deal with sources of unwanted variation such as batch effects and
  unmeasured confounders.  We describe a simple modular fitting
  procedure that combines key ideas from both these lines of
  research. This yields new, powerful EB methods for analyzing
  genomics experiments that account for both sparse effects and
  unwanted variation.  In realistic simulations, these new methods
  provide significant gains in power and calibration over competing
  methods. In real data analysis we find that different methods, while
  often conceptually similar, can vary widely in their assessments of
  statistical significance. This highlights the need for care in both
  choice of methods and interpretation of results. All methods
  introduced in this paper are implemented in the R package
  \texttt{vicar} available at \url{https://github.com/dcgerard/vicar}.
\end{abstract}

\let\thefootnote\relax
\addtocounter{footnote}{-1}\let\thefootnote\svthefootnote

\section{Introduction}
\label{section:introduction}

Many modern genomics experiments involve scanning the genome, or a
list of genomic units (e.g.~``genes''), to detect differences between
groups of samples.  For example, a simple ``differential expression''
experiment might measure the expression (activity level) of many genes
in samples from two groups, and aim to identify at which genes these
groups differ in their mean expression. The motivation is that identifying
such genes may yield insights into the biological basis of
differences between the groups.

Analyses of such experiments involve many issues, but two are
particularly important and arise repeatedly. The first is that effects
are often sparse --- for example, in a differential expression
experiment, many genes may show little difference in expression
between two groups.  The second is that genomic experiments are often
plagued by ``unwanted variation'' such as batch effects and unmeasured
confounders
\citep{leek2007capturing,leek2008general,gagnon2012using}. It is
crucial to address both these issues during statistical analyses. The
sparsity of effects requires careful handling of statistical
significance thresholds to avoid large numbers of false
discoveries. And unwanted variation, if unaccounted for, can obscure
or confound signals of interest, and can create the appearance of
signals where they do not exist.

Here we combine two ideas that have been used to address these
issues. The first is the use of Empirical Bayes (EB) methods to assess
the sparsity of effects, and estimate false discovery rates (FDRs) and
related quantities
\citep[e.g][]{efron2004large,efron2008microarrays,stephens2016false}.
The second is the use of factor analysis (FA) to deal with sources of
unwanted variation such as batch effects and unmeasured confounders
\citep[e.g.][]{lucas2006sparse, leek2007capturing, leek2008general,
  gagnon2012using, sun2012multiple, gerard2017unifying,
  wang2017confounder}. By combining ideas from both these lines of
research we provide powerful new analysis methods that simultaneously
account for both sparse effects and unwanted variation.

Our work is not the first to combine sparsity of effects with FA
for unwanted variation. Indeed ``Fully Bayesian'' approaches that do
this were among the first work on dealing with unwanted variation
\citep[e.g.][]{lucas2006sparse,carvalho2008high}. However, these
methods are complex, computationally challenging, and have not been
widely adopted, perhaps in part because of lack of easy-to-use
software implementations. In comparison our EB methods are relatively
simple, and we provide implementations in an R package.  Also, our EB
methods exploit recently-introduced semi-parametric prior
distributions \citep{stephens2016false} which assume that the
distribution of effects is unimodal at zero.  These priors are both
computationally convenient, and more flexible than those used in
previous Bayesian work.

More recently, \citet{sun2012multiple} and \citet{wang2017confounder}
introduced (non-Bayesian) approaches that combine sparsity of effects
and FA for unwanted variation.  Indeed \citet{wang2017confounder} give
theory that supports combining these two ideas:~the
estimation of the effects and the FA are intimately entwined, and
assuming sparsity of effects helps identify the unwanted variation.
To implement this idea \citet{wang2017confounder} --- building
directly on \citet{sun2012multiple} --- jointly estimate the effects
and the unwanted variation, using a penalty to induce sparsity on the
effects. Our work here takes a similar approach, but replaces the
penalty approach with EB methods to induce sparsity.  The EB approach
has several advantages over a penalized approach:~it provides not only
sparse point estimates, but also shrunken interval estimates, and
estimates of FDRs and related quantities. And the semi-parametric
prior distributions we use are considerably more flexible than the
penalty approach (which often has only a single parameter to control
sparsity and shrinkage).

Our methods based on assuming sparse (or, more precisely, unimodal)
effects provide an attractive alternative to methods based on
``control genes'' \citep{gagnon2012using}, which are genes assumed
{\it a priori} to have no effect.  Like the sparsity assumption, the
control gene assumption helps identify the unwanted variation
\citep{gagnon2012using, wang2017confounder}.  However, while the
sparsity assumption is almost universally adopted in genomic analyses
(implicitly or explicitly), the control gene assumption brings a
considerable additional burden:~specifying a suitable set of controls
is non-trivial and potentially error-prone.  Furthermore, even when
the controls are perfectly chosen, our methods can produce better
results, particularly if the number of control genes is small (see
Section \ref{section:empirical.evaluations}).  (It would be
straightforward to incorporate control genes --- as well as sparsity
--- into our method, but we do not pursue this here.)

One key feature of our method (also shared by several methods
mentioned above) is its ``modularity''.  In particular we exploit a
modular fitting procedure \citep[e.g.][]{wang2017confounder} that {\it
  jointly} estimates the effects and FA, while also {\it separating
  out the FA} from the rest of the method. Consequently, no particular
approach to FA is ``baked in'' to our method; instead it can easily
accommodate any approach to FA, including for example Bayesian
approaches to FA \citep[e.g.][]{hoff2007model, stegle2008accounting,
  engelhardt2010analysis, stegle2010bayesian}. Similarly, the method
can accommodate a range of possible pre-processing steps that are
often necessary in genomic data analysis.  The modular approach also
simplifies computation, and eases both implementation and
interpretation. Indeed our methods maintain much of the simple modular
structure and logic of the simplest existing approaches to this
problem.  The benefits of modularity, while widely recognized in
software design, are rarely explicitly acknowledged in statistical
methods development, and we believe they merit greater recognition.

On notation:~we generally denote matrices by uppercase boldface
($\bs{A}$), vectors by lowercase boldface ($\bs{a}$), scalars by
lowercase non-boldface ($a$), and sets with calligraphic letters
($\mathcal{A}$). There are exceptions when the context is clear. For
example $\bs{\beta}$ is sometimes a matrix and sometimes a
vector. Elements of a vector or matrix are denoted by their lowercase
non-boldface versions. For example $a_i$ is the $i$th element of
$\bs{a}$ and $a_{ij}$ is the $(i,j)$th element of $\bs{A}$. We let
$\bs{A}_{n \times p}$ denote that the matrix $\bs{A}$ has dimension
$n \times p$, i.e.\ $\bs{A} \in \mathbb{R}^{n \times p}$.

\section{Background} \label{section:background}

\subsection{Typical analysis pipeline} \label{section:genomics}

Genomics researchers often aim to identify which genomic features are
associated with one or more biological factors of interest.  For
example, which genes have activity levels that differ, on average,
between males and females?  To assess this they would measure gene
expression at many genes on samples of each sex, and then perform
statistical analyses to identify which genes show significant
differences in mean levels between the two groups.

There are many ways to perform such statistical analyses
\citep[e.g.][]{soneson2013comparison}, but in outline a typical
analysis might involve:
\begin{enumerate}[noitemsep, nolistsep]
\item For each gene, $j$, estimate an effect size $\bhat_j$ and
  corresponding standard error $\shat_j$. (In our example $\bhat_j$
  would represent the estimated difference in mean gene expression
  between the two sexes.) For example, this might be achieved by
  applying a linear model to appropriately normalized and/or
  transformed expression data \citep[e.g.][]{law2014voom}, combined
  with methods to moderate (stabilize) variance estimates
  \citep[e.g.][]{smyth2004linear}.
\item For each gene, $j$, use $\bhat_j, \shat_j$ to obtain a
  $p$-value, $p_j$, testing the null hypothesis that gene $j$ shows no
  effect. Then apply FDR methods
  \citep{benjamini1995controlling,storey2003positive} to the set of
  all $p$ values to decide which genes are ``significant''.
\end{enumerate}

\subsection{Adaptive shrinkage} \label{section:ash}

Building on ideas of \citet{efron2004large,efron2008microarrays},
\citet{stephens2016false} suggests an alternative to Step 2 above,
which he calls ``adaptive shrinkage'' or \ash{}.  Specifically,
\citet{stephens2016false} suggests combining the ``observations''
$(\vbhat,\vshat)$ from Step 1 with a flexible but powerful
assumption:\ that the true effects ($\vbeta$) come from a unimodal
distribution with mode at 0. This assumption captures the expectation
that many effects will be at or near 0, and is effectively an analogue
of (or alternative to) the ``sparsity assumption'' often made in this
context. \citet{stephens2016false} provides methods to estimate this
unimodal distribution, and to compute posterior distributions and
measures of significance for each effect --- the local FDR (lfdr;
\citet{efron2008microarrays}), and local false sign rate (lfsr;
\citet{stephens2016false}) --- analogous to the standard pipeline
above.  \citet{stephens2016false} highlights several advantages of this
approach:\ it better accounts for differences in measurement precision
($\shat_j$) among genes; it can provide better (less conservative)
estimates of the FDR, provided the unimodal assumption holds; and it
provides calibrated interval estimates for each effect, which are
otherwise difficult to obtain.

In more detail:\ \ash{} uses the normal means model
\citep{stein1981estimation} to relate the observations
$(\vbhat, \vshat)$ to the effects $\vbeta$:
\begin{align}
  \label{eq:normal.lik.model}
\vbhat \, | \, \bs{\beta}, \vshat \sim N_p(\bs{\beta},\bs{S}),
\end{align}
where $N_p$ denotes the $p$-dimensional multivariate normal
distribution and $\bs{S} := \diag(\shat_1^2,\ldots,\shat_p^2)$.  Thus
the likelihood for $\vbeta$ is
\begin{align} \label{eq:normal.lik}
L(\vbeta; \vbhat, \vshat) = \prod_{j = 1}^p N(\bhat_j|\beta_j, \shat_j^2),
\end{align}
where $N(\cdot|a, b^2)$ denotes the normal density function with mean
$a$ and variance $b^2$.  This likelihood is then combined with the
unimodal assumption:
\begin{align}
\label{eq:beta.prior.spec}
\beta_1,\ldots,\beta_p \overset{iid}{\sim} g \in \mathcal{U},
\end{align}
where $\mathcal{U}$ denotes the space of unimodal distributions with
mode at 0.

\citet{stephens2016false} provides methods to fit the model
\eqref{eq:normal.lik}-\eqref{eq:beta.prior.spec} using a two-step EB
approach:
\begin{enumerate}[noitemsep, nolistsep]
\item Estimate $g$ by maximizing the marginal likelihood:
  \begin{align}
    \label{eq:marginal.density}
    \hat{g} = \argmax_{g\in\mathcal{U}}p(\bhat|g, \vshat) = \argmax_{g\in\mathcal{U}}\prod_{j = 1}^p\int_{\beta_j}N( \bhat_j|\beta_j, \shat_j^2) g(d\beta_j).
  \end{align}
\item Compute posterior distributions
  $p(\beta_j | \hat{g}, \vbhat, \vshat)$, and return posterior
  summaries, including the lfdr and lfsr.
\end{enumerate}

In practice, \ash{} approximates the optimization
\eqref{eq:marginal.density} by exploiting the fact that any unimodal
distribution can be approximated arbitrarily well using a finite
mixture of uniform distributions.  Using this representation,
\eqref{eq:marginal.density} becomes a convex optimization problem over
a finite (but large) set of mixture weights
$\bs{\pi} = (\pi_1,\ldots,\pi_M)$ (see equation
\eqref{eq:prior.mixture} later).  This can be solved efficiently using
interior point methods \citep{boyd2004convex, koenker2014convex}.

\subsection{Removing Unwanted Variation}
\label{subsection:rotate}

Unwanted variation can plague genomics experiments that aim to
identify systematic differences in gene expression, or other genomics
features, among groups of samples \citep[Appendix
\ref{section:simple.illustration} of the Supplementary
Materials]{leek2008general, leek2007capturing, stegle2010bayesian,
  leek2010tackling, gagnon2012using, sun2012multiple}.  Unwanted
variation may include measured variables such as batch, or sample
covariates such as age or sex, but also --- and most challengingly ---
unmeasured variables, such as aspects of sample preparation and
handling that may be difficult to measure and control.  Unwanted
variation, if unaccounted for, can obscure or confound signals of
interest, and can create the appearance of signals where they do not
exist.

As the severity of the problems caused by unwanted variation has been
increasingly recognized, many statistical methods have been developed
to help ameliorate them \citep{lucas2006sparse, leek2007capturing,
  sun2012multiple, gagnon2013removing, gerard2017unifying,
  wang2017confounder}.  Most of these methods are based on a
``factor-augmented regression model''
\citep{leek2007capturing,leek2008general}:
\begin{align}
  \label{eq:full.model}
  \bs{Y}_{n\times p} = \bs{X}_{n \times
    k}\bs{\beta}_{k \times p} + \bs{Z}_{n \times
    q}\bs{\alpha}_{q \times p} + \bs{E}_{n\times
    p},
\end{align}
where $y_{ij}$ is the normalized expression level of gene $j$ in
sample $i$; $\bs{X}$ is a matrix containing observed covariates, with
$\bs{\beta}$ a matrix of corresponding effects; $\bs{Z}$ is a matrix
of unobserved factors causing unwanted variation, with $\bs{\alpha}$ a
matrix of corresponding effects; and $\bs{E}$ has independent
(Gaussian) errors with means 0 and column-specific variances
$\var(e_{ij}) = \sigma_j^2$. In \eqref{eq:full.model} only $\bs{Y}$
and $\bs{X}$ are known; other quantities are to be estimated.

Here we focus on the common setting where only one of the covariates
in the columns of $\bs{X}$ is of interest, and the other $k - 1$
covariates are included to improve the model (e.g.~to control for
measured confounders, or as an intercept term). To further simplify
notation we focus on the case $k=1$, so $\bs{X}$ is an $n$-vector, and
$\vbeta$ is a $p$-vector of the effects of interest. However, our
methods and software implementation allow $k>1$. See Appendix
\ref{section:detailed.review} of the Supplementary Materials for
details. See also Appendix \ref{section:linear.combo} of the
Supplementary Materials where we further discuss how to apply these
methods when a single linear combination of the effects are of
interest.

There are many approaches to fitting \eqref{eq:full.model}. Here we
exploit a modular approach used by several previous methods, including
RUV4 \citep{gagnon2013removing}, LEAPP \citep{sun2012multiple}, and
CATE \citep{wang2017confounder}. In outline this involves:
\begin{enumerate}[noitemsep, nolistsep]
\item For each gene $j$ ($j=1,\dots,p$) obtain an initial estimate
  $\bhat_j$ for $\beta_j$ ignoring unwanted variation by using
  ordinary least squares (OLS) regression of the $j$th column of
  $\bs{Y}$ on $\bs{X}$.
\item Form the matrix of residuals from these regressions,
  $\tilde{\bs{Y}}:= \bs{Y} - \bs{X}\vbhat$, and perform a FA on these
  residuals. (Some methods, including CATE and the methods we present
  here, perform this step in practice by applying FA to a slightly
  different matrix. However, the end result is similar or identical,
  and we find it simpler and more intuitive to describe the methods in
  terms of the residual matrix. See Appendix
  \ref{section:detailed.review} of the Supplementary Materials for
  details.) Performing an FA on $\tilde{\bs{Y}}$ means fitting a model
  of the form:
\begin{equation} \label{eqn:FA}
\tilde{\bs{Y}} = \tilde{\bs{Z}} \tilde{\bs{\alpha}}  + \tilde{\bs{E}}.
\end{equation}
Most methods are flexible about exactly how FA is performed here, at
least in principal if not in software.  The resulting estimate
$\hat{\bs{\alpha}}$ of $\tilde{\bs{\alpha}}$ in \eqref{eqn:FA} can be
viewed as an estimate of $\bs{\alpha}$ in \eqref{eq:full.model}.  This
step also yields estimates $\hat{\sigma}^2_j$ of the residual
variances $\sigma^2_j$ in \eqref{eq:full.model}.
\item Estimate $\bs{\beta}$ by jointly estimating
  $(\bs{\beta},\bs{z})$ in the following ``simplified model'':
\begin{align}
  \label{eq:simplified.model}
  \vbhat \sim N_{p}(\vbeta + \hat{\bs{\alpha}}^{\intercal}\bs{z}, \bs{S}),
\end{align}
where $\bs{z} \in \mathbb{R}^{q}$, $\vbhat \in \mathbb{R}^{p}$ are the
OLS estimates from Step 1, and
$\bs{S} = \diag(\shat_1^2,\ldots,\shat_p^2)$ where $\shat_j$ is an
estimated standard error of $\bhat_j$,
\begin{align}
\label{eq:shat.def}
\shat^2_j = \hat{\sigma}^2_j/(\bf{X}^T \bf{X}).
\end{align}
Model \eqref{eq:simplified.model} has a simple interpretation:~the OLS
estimates $\vbhat$ are equal to the true coefficients ($\vbeta$) plus
a bias term due to unwanted variation
($\hat{\bs{\alpha}}^{\intercal}\bs{z}$) plus some noise
($N_p(\bs{0}, \bs{S})$). That is $\bs{z}$ can be interpreted as
capturing the effect of the unwanted variation on the OLS estimates.
\end{enumerate}
This modular approach to fitting the model \eqref{eq:full.model} is
less {\it ad hoc} than it may first seem, and can be rigorously
justified (\citet{wang2017confounder}; see Appendix
\ref{section:detailed.review} of the Supplementary Materials for a
detailed review).

A key way in which methods differ is the assumptions they make when
fitting model \eqref{eq:simplified.model}.  This model contains
$p + q$ parameters but only $p$ observations, so additional
assumptions are clearly necessary \citep{wang2017confounder}.

One type of method assumes that some genes are ``control genes''
\citep{gagnon2013removing}.  That is, to assume that for some set
$\mathcal{C} \subseteq \{1,\ldots, p\}$, the effects $\beta_{j} = 0$
for all $j \in \mathcal{C}$. For these control genes
\eqref{eq:simplified.model} becomes:
\begin{align}
\label{eq:reduced.model.control}
\hat{\vbeta}_{\mathcal{C}} \sim N_p(\hat{\bs{\alpha}}_{\mathcal{C}}^{\intercal}\bs{z}, \bs{S}_{\mathcal{C}}),
\end{align}
where $\hat{\vbeta}_{\mathcal{C}}$ denotes the elements of
$\hat{\vbeta}$ that correspond to indices in $\mathcal{C}$. Fitting
this model yields an estimate for $\bs{z}$, $\hat{\bs{z}}$, say.
Substituting this estimate into \eqref{eq:simplified.model} then
yields an estimate for $\vbeta$,
\begin{equation} \label{eq:bhat.adj}
\vbhat' = \vbhat - \hat{\bs{\alpha}}^{\intercal}\hat{\bs{z}}.
\end{equation}
This approach is used by both RUV4 and the negative controls version
of CATE (CATEnc), with the difference being that RUV4 uses OLS when
estimating $\bs{z}$ whereas CATEnc uses generalized least squares
(GLS).

An alternative approach, used by LEAPP \citep{sun2012multiple} and the
robust regression version of CATE (CATErr) \citep{wang2017confounder},
is to assume the effects $\vbeta$ are sparse.  Both LEAPP and CATErr
do this by introducing a penalty on $\vbeta$ when fitting
\eqref{eq:simplified.model}.  LEAPP returns the estimates of $\vbeta$
from this step (so these estimates are sparse and/or shrunk due to the
sparsity-inducing penalty). CATErr, instead only keeps the estimates
of $\bs{z}$ and estimates $\vbeta$ by \eqref{eq:bhat.adj}.  Our
methods here essentially involve replacing the sparsity-inducing
penalty with the unimodal assumption from \ash{}.

\section{MOUTHWASH}
\label{section:mouthwash}

Here we combine the EB method from \ash{} with the modular fitting
procedure for removing unwanted variation outlined above.  This yields
an analysis pipeline that combines the benefits of \ash{} (see above)
while also removing unwanted variation.  In brief, our new method
involves replacing the likelihood \eqref{eq:normal.lik.model} in
\ash{} with the likelihood \eqref{eq:simplified.model}, which accounts
for unwanted variation.  We then modify the EB approach of \ash{} to
optimize over both the unimodal prior distribution $g$ and the
unwanted variation $\bs{z}$.  We call this method MOUTHWASH
(\textbf{M}aximizing \textbf{O}ver \textbf{U}nobservables \textbf{T}o
\textbf{H}elp \textbf{W}ith \textbf{A}daptive \textbf{SH}rinkage).

In more detail, MOUTHWASH involves:
\begin{enumerate}[noitemsep, nolistsep]
\item Estimate effects $\bhat_j$ by OLS regression of the $j$th column
  of $\bs{Y}$ on $\bs{X}$.
\item Obtain $\hat{\bs{\alpha}}$ and $\hat{\sigma}_j$ by
  applying a FA to the residual matrix $\tilde{\bs{Y}}$. (These first
  two steps are the same as RUV4, LEAPP and CATE, as outlined above.)
\begin{description}
\item 2b. Optionally, apply variance moderation
  \citep{smyth2004linear} to adjust the $\hat{\sigma}_j$'s \citep[as
  in][]{gagnon2013removing}. We do this using $n - k - q$ as the
  degrees of freedom.
\end{description}
\item Estimate the unimodal effects distribution $g$ and the unwanted
  variation effects $\bs{z}$ by maximum (marginal) likelihood applied
  to \eqref{eq:simplified.model}:
  \begin{align}
    \label{eq:mouth.opt}
    \begin{split}
      (\hat{g},\hat{\bs{z}}) &:= \argmax_{(g,\bs{z})\ \in\ \mathcal{U}\times \mathbb{R}^k}p(\vbhat |g,\bs{z},\hat{\bs{\alpha}},\vshat)\\
      &= \argmax_{(g,\bs{z})\ \in\ \mathcal{U}\times \mathbb{R}^k}\prod_{j = 1}^p \int_{\beta_j}N(\bhat_j|\beta_j + \hat{\bs{\alpha}}_j^{\intercal}\bs{z}, \shat_j^2)g(d\beta_j),
    \end{split}
  \end{align}
  where $\hat{s}_j$ is defined in \eqref{eq:shat.def}.
\item Compute posterior distributions
  $p(\beta_j | \hat{g}, \hat{\bs{z}}, \vbhat, \vshat)$, and return
  posterior summaries.
\end{enumerate}

The key new step is Step 3. As in \cite{stephens2016false} we approximate this
optimization by optimizing $g$ over a set of finite mixture
distributions indexed by mixing proportions $\bs{\pi}$:
\begin{align}
\label{eq:prior.mixture}
g(\beta_j|\bs{\pi}) &= \pi_0\delta_0(\beta_j) + \sum_{m = 1}^M\pi_mf_m(\beta_j),
\end{align}
where the $f_k$ are pre-specified component pdf's with one of the following forms:
\begin{enumerate}[noitemsep, nolistsep]
\item[i)] $f_m(\cdot) = N(\cdot|0,\tau_m^2)$,
\item[ii)] $f_m(\cdot) = U[\cdot|-a_m,a_m]$,
\item[iii)] $f_m(\cdot) = U[\cdot|-a_m,0]$ or $U[\cdot|0,a_m]$,
\end{enumerate}
where $U[\cdot|a, b]$ denotes the uniform density with lower limit $b$
and upper limit $a$. These three different options correspond
respectively to (approximately) optimizing $g$ over i) all
(zero-centered) scale mixtures of normals; ii) symmetric unimodal
distributions with mode at 0; iii) all unimodal distributions with
mode at 0.

With this mixture representation the integral in \eqref{eq:mouth.opt}
can be computed analytically, and optimization can be performed using
either an EM algorithm (Appendix \ref{section:em.normal} of the
Supplementary Materials) or a coordinate ascent algorithm (Appendix
\ref{section:em.uniform} of the Supplementary Materials). Although
this optimization problem is --- in contrast to \ash{} --- no longer
convex, we have found that with appropriate initialization of
$\bs{\pi}$ (initializing $\pi_0$ close to 1) these algorithms produce
consistently reliable results (Supplementary Figure
\ref{figure:pi0.5}). Thus, for each simulated and real dataset we run
MOUTHWASH once from this initialization.

\subsection*{Identifiability}

In \eqref{eq:full.model}, as in any factor model, identifiability
issues arise.  Specifically, the following likelihoods are equivalent:
\begin{align}
\label{eq:not.ident}
p(\bs{Y} | \bs{\beta}, \bs{Z}\bs{A}, \bs{A}^{-1}\bs{\alpha}, \bs{\Sigma}) =
p(\bs{Y} | \bs{\beta}, \bs{Z}, \bs{\alpha}, \bs{\Sigma}),
\end{align}
for any non-singular $\bs{A} \in \mathbb{R}^{q \times q}$.  The result
of this non-identifiability is that (in the absence of prior
information on $\bs{\alpha}$) the estimate of $\bs{\alpha}$ from Step
2 above can be considered identified only up to its rowspace.  It
therefore seems desirable that the estimates obtained in Steps 3 and 4
of MOUTHWASH should depend on $\hat{\bs{\alpha}}$ \emph{only} through
its rowspace. \citet{gagnon2013removing} proved that their estimator
satisfied this property. We prove in Theorem \ref{theorem:row.space}
(Appendix \ref{section:identifiability} of the Supplementary
Materials) that our estimator also satisfies this property.

\subsection{Errors in variance estimates}
\label{section:var.inflate}

The performance of MOUTHWASH (and other related methods) depends on
obtaining accurate variance estimates $\hat\sigma_j$ in Step 2.  In
practice this can be a major problem. See for example Section 3.9.4 of
\citet{gagnon2013removing}, Section 6 of \citet{gerard2017unifying},
and \citet{perry2013degrees} (who consider a similar model to
\eqref{eq:full.model} with the assumption that the unobserved factors
are orthogonal to the observed covariates).  Intuitively, the
difficulty may arise either from mispecifying the number of latent
factors and thus attributing either too much or too little variation
to the noise \citep{gagnon2013removing}; or it may arise because
$\hat{\bs{\alpha}}$ is assumed known but is in fact estimated and so
the variance in the assumed model \eqref{eq:simplified.model} is too
small.

Both \citet{gagnon2013removing} and \citet{perry2013degrees} address
this issue by applying a multiplicative factor to the variance
estimates. (\citet{gagnon2013removing} selects this factor using
control genes, whereas \citet{perry2013degrees} selects this factor
via asymptotic arguments.)  Here we deal with this issue in a similar
way by including a multiplicative parameter, $\xi > 0$ in
\eqref{eq:simplified.model}.

Specifically, we modify \eqref{eq:simplified.model} to:
\begin{align}
\label{eq:model.var.inflate}
\hat{\vbeta} \sim N_{p}(\bs{\beta} + \hat{\bs{\alpha}}^{\intercal}\bs{z}, \xi\bs{S}),
\end{align}
and estimate $\xi$ along with $g$ and $\bs{z}$. Thus, Step 3 becomes:
 \begin{align}
    \label{eq:mouth.opt.var}
    \begin{split}
      (\hat{g},\hat{\bs{z}}, \hat{\xi}) &= \argmax_{(g,\bs{z}, \xi)\ \in\ \mathcal{U}\times \mathbb{R}^k \times \mathbb{R}^+}\prod_{j = 1}^p \int_{\beta_j}N(\hat{\beta}_j|\beta_j + \hat{\bs{\alpha}}_j^{\intercal}\bs{z}, \xi \shat_j^2)g(\beta_j)\dif\beta_j,
    \end{split}
  \end{align}
  and the posterior distributions in Step 4 are computed conditional
  on $\hat{\xi}$.  We have found that this modification can be vital for
  good performance of MOUTHWASH in practice.

\subsection{Other Bells and Whistles}
\label{section:bells.whistles}

We have implemented several extensions to this approach in our
software. These include i) allowing effects to depend on their
standard errors; ii) extending \eqref{eq:simplified.model} to a $t$
likelihood; iii) introducing a small regularization on the mixing
proportions in $g$ to promote conservative behavior; and iv) reducing
computational burden when $p$ is large by subsampling of genes. These
are described in Appendix \ref{section:additional.bells} of the
Supplementary Materials. (In our practical illustrations here we use
the regularization iii), but not the other features.)

Additionally, to better account for the uncertainty in estimating
$\bs{z}$, we implemented a related procedure called BACKWASH
(\textbf{B}ayesian \textbf{A}djustment for \textbf{C}onfounding
\textbf{K}nitted \textbf{W}ith \textbf{A}daptive \textbf{SH}rinkage)
that places a prior over $\bs{z}$. See Appendix
\ref{section:backwash} of the Supplementary Materials for details.

\section{Empirical Evaluations} \label{section:results}
\label{section:empirical.evaluations}

\subsection{Simulations}
\label{section:simulations}

To compare methods we generated simulated datasets from experimental
data that contain real unwanted variation. Specifically, following
\citet{gerard2017unifying}, we simulated data by first randomly
partitioning real RNA-seq data into two groups to produce ``null''
data, and then modifying it to spike in known amounts of signal.  In
brief, we modify the RNA-seq counts at a randomly selected subset of
genes by ``thinning'' the RNA-seq data, reducing the RNA-seq counts in
one group or the other to make each gene systematically less expressed
in that group.  See Appendix \ref{section:signal.details} of the
Supplementary Materials for details.

Because these simulations start by randomly assigning group labels to
samples, they mimic a randomized experiment where unwanted variation
is independent of treatment. In this sense they represent a
``best-case'' scenario, but with realistic, challenging, levels of
unwanted variation.  Although any simulation is inevitably a
simplification, we believe that these simulations provide a
substantially better guide to method performance in practice than
simulating under an assumed (and undoubtedly imperfect) model.

We used these simulations to compare MOUTHWASH and BACKWASH with nine
other estimation procedures that we follow with either qvalue
\citep{storey2003positive} or \ash{} to estimate FDRs. (Although,
based on \citet{stephens2016false}, we would advocate using the lfsr
rather than FDR or lfdr, here we use FDR to allow comparison with
methods that do not compute the lfsr.)  These nine estimation methods
are:
\begin{enumerate}[noitemsep, nolistsep]
\item OLS:~Ordinary Least Squares. This represents a naive method that
  does not account for unwanted variation.
\item SVA:~The iteratively re-weighted least-squares version of
  Surrogate Variable Analysis \citep{leek2008general}, followed by the
  widely-used ``voom-limma" pipeline \citep{law2014voom} to obtain
  effect estimates and standard errors controlling for the estimated
  surrogate variables.
\item CATErr:~The robust regression version of CATE
  \citep{wang2017confounder} (a variation on LEAPP
  \citep{sun2012multiple}).
\item CATErr+MAD:~CATErr, followed by median centering and median
  absolute deviation (MAD) scaling of the $t$-statistics
  \citep{sun2012multiple, wang2017confounder}. At time of writing this
  was the default option in the {\tt cate} package.  (When applying
  \ash{}, we used the MAD as a multiplicative factor to adjust the
  variances \citep{gerard2017unifying}, rather than scaling the $t$
  statistics.)
\item RUV2 \citep{gagnon2012using}.
\item RUV3 \citep{gerard2017unifying}, with EB variance moderation
  \citep{smyth2004linear}.
\item CATEnc:~the negative controls version of CATE
  \citep{wang2017confounder} (a variant on RUV4
  \citep{gagnon2013removing}), which uses control genes to help
  estimation of confounders.
\item CATEnc+MAD:~CATEnc followed by the same standardization used in
  CATErr+MAD.
\item CATEnc+Cal:~CATEnc where a multiplicative factor, calculated
  using control genes \citep{gagnon2013removing}, was used to adjust
  the variances .
\end{enumerate}
The last five of these methods (RUV2, RUV3, CATEnc, CATEnc+MAD,
CATEnc+Cal) require control genes, and we provided them a random
subset of the actual null genes as controls, again representing a
``best case'' scenario for these methods. We did not adjust for
library size in any method as library size can be considered another
source of unwanted variation \citep{gerard2017unifying}, which these
methods are designed to account for.

We performed simulations with $p=1000$ genes, varying the following
parameters:
\begin{itemize}[noitemsep, nolistsep]
\item The proportion of genes that are null $\pi_0 \in \{0.5, 0.9, 1\}$,
\item The number of samples $n \in \{6, 10, 20, 40\}$,
\item The number of control genes provided to methods that use control
  genes $m \in \{10, 100\}$.
\end{itemize}
We simulated 500 datasets for each combination of $\pi_0$, $n$, and
$m$, and ran all methods on each dataset. We evaluated performances
based on two criteria:~first, the area under their receiver operating
characteristic curve (AUC), a measure of their ability to distinguish
null versus non-null genes; and second, accuracy of estimated
proportion of null genes ($\pi_0$), which is an important step in
providing calibrated FDR estimates.

Figure \ref{figure:auc} compares the AUCs of each method. MOUTHWASH
and BACKWASH have almost identical performance, and the best AUCs in
almost every scenario (SVA methods have better AUC in small sample
sizes with $\pi_0 = 0.5$). This dominance is particular pronounced
when the number of control genes is small ($m = 10$), where methods
that use control genes falter. With $m=100$ high-quality control
genes, methods that use control genes become competitive with
MOUTHWASH and BACKWASH.

\begin{figure}
\begin{center}
\includegraphics[scale = 0.9]{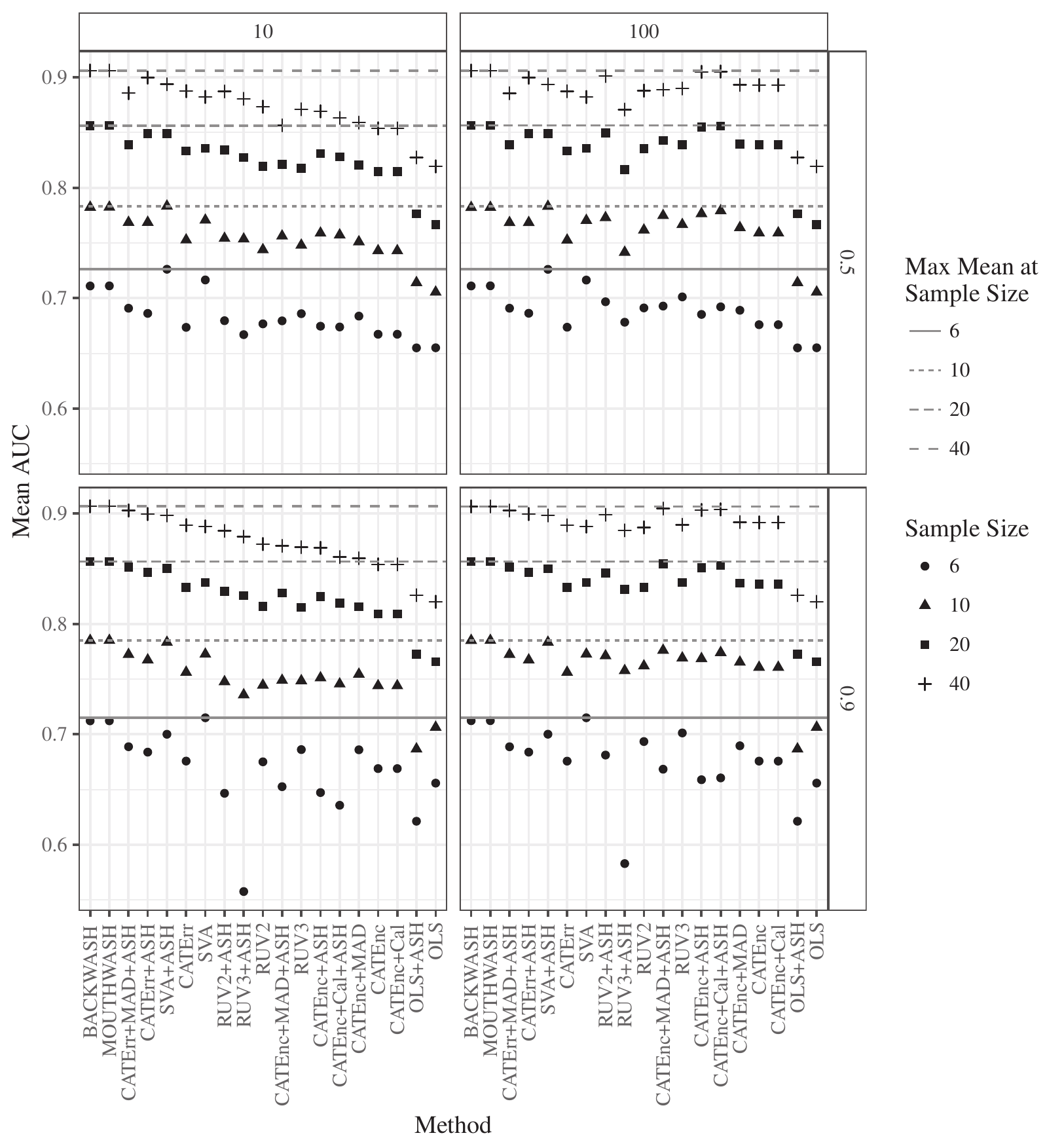}
\end{center}
\caption{Comparison of mean AUCs among methods. Column facets vary
  $m$, the numbers of control genes made available to methods that use
  control genes. Row facets vary $\pi_0$, the proportions of null
  genes. Different symbols represent different sample sizes
  $n$. Horizontal lines indicate the highest mean AUC achieved by any
  method at a given combination of sample size, number of control
  genes, and proportion of null genes. The methods are ordered by
  their performance in the simulations with $n=40, m=10,\pi_0= 0.9$.}
\label{figure:auc}
\end{figure}

Figure \ref{figure:pi0.9} compares the estimates of $\pi_0$ for each
method when the true $\pi_0=0.9$ (results for $\pi_0 = 0.5$ and $1$
are in Supplementary Figures \ref{figure:pi0.5} and
\ref{figure:pi0.1}). Many methods have median estimates of $\pi_0$
very close to the true value of $0.9$. However, the variances of these
estimates are often high. In comparison, the estimates of $\pi_0$ from
MOUTHWASH and BACKWASH are much less variable, and hence more accurate
on average (particularly at higher sample sizes).  CATErr+MAD+ASH and
CATEnc+MAD+ASH work very well for larger sample sizes when $\pi_0$ is
close to 1, but are anti-conservative for small sample sizes and highly
conservative when $\pi_0 = 0.5$ (Supplementary Figure
\ref{figure:pi0.5}).  Results from MOUTHWASH and BACKWASH are almost
identical, suggesting that the additional complexity of BACKWASH is
unnecessary in practice.

\begin{figure}
\begin{center}
\includegraphics[scale = 0.9]{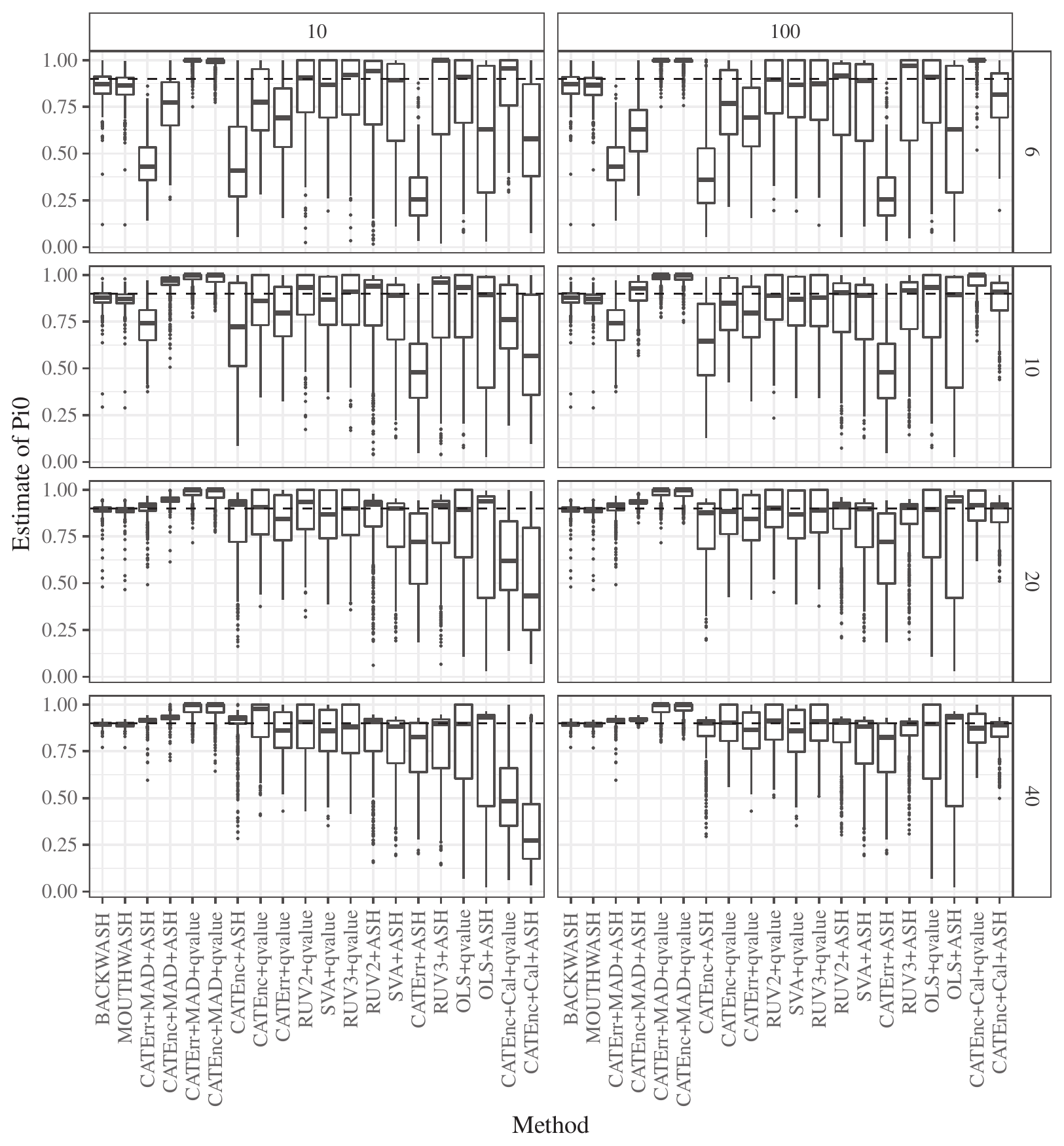}
\end{center}
\caption{Boxplots of estimates of $\pi_0$ for each method (true
  $\pi_0 = 0.9$). Column facets vary $m$, the numbers of control genes
  made available to methods that use control genes. Row facets vary
  $n$, the sample size. The methods are ordered by the mean squared
  error of their estimates in the simulations with
  $n=40, m=10, \pi_0=0.9$. The dashed horizontal line shows $y=0.9$.}
\label{figure:pi0.9}
\end{figure}

\subsection{Computation Time}

Although our MOUTHWASH method is significantly slower than other
existing methods (see Table \ref{tab:comp.time} in the Supplementary
Materials), it is nonetheless practical for realistic-sized data.  For
example, in tests with $n = 100$ and $p = 10\text{,}000$ MOUTHWASH had
a median runtime of 140 seconds (on a 4.0 GHz quad-core PC running
Linux with 32 GB of memory), and runtime is similar for other values
of $n$.  Further speedups could be achieved if needed; see Appendix
\ref{section:additional.bells} of the Supplementary Materials for
discussion.  BACKWASH requires a significantly longer runtime than
MOUTHWASH, and given their similar performance we prefer MOUTHWASH in
practice.

\subsection{GTEx Data}
\label{section:real.data}

To evaluate methods on real data, \citet{gagnon2012using} used the
idea of positive controls. A positive control is a gene that is
\emph{a priori} thought likely to be associated with the covariate of
interest. \citet{gagnon2012using} used the example of sex and sex
chromosomes:~when the covariate of interest is the sex of an
individual, then the genes on sex chromosomes are positive
controls. The best confounder adjustment methods, then, are those that
tend to have more positive controls among their most significant
genes. This idea is also used in \citet{gagnon2013removing} and
\citet{wang2017confounder}.

We applied this positive control method using RNA-seq datasets from 23
non-sex-specific tissues collected by the GTEx project
\citep{gtex2015}. In each dataset we filtered out low-expressed genes
(mean expression level $<$ 10 reads), applied a $\log_2$
transformation to the gene expression count matrix (after adding a
pseudo-count), and averaged results over technical replicates. We used
a design matrix $\bs{X}\in\mathbb{R}^{n \times 2}$ with two columns:~a
column of $1$'s (intercept), and a column of indicators for sex.  We
applied the same methods as in Section \ref{section:simulations} to
all 23 datasets.  For methods that require negative controls we
followed \citet{gagnon2012using} in using housekeeping genes as
negative controls (although opinions seem divided on the general
appropriateness of this strategy;~see \citet{zhang2015do} for a
detailed discussion). Specifically, we used the list of housekeeping
genes from \citet{eisenberg2013human}, but excluding sex-chromosome
genes. (A newer list of housekeeping genes was released by
\citet{lin2017housekeeping} based on single cell sequencing
results. We repeat our following analyses in Appendix
\ref{section:lin} of the Supplementary Materials using this newer
list. The results of Appendix \ref{section:lin} are similar to those
obtained here.)

To compare methods we took the most significant 100 genes for each
method on each tissue and counted how many of these genes are on a sex
chromosome ($s$). We divided $s$ for each method by the maximum $s$
among all methods within a tissue.  Figure \ref{figure:prop.max} shows
the results, with white indicating better performance (larger
$s$). Methods are ordered from left to right by their median
performance.  Overall most methods performed comparably, CATEnc
variants and SVA the notable exceptions, with SVA performing
particularly poorly on a subset of the tissues. MOUTHWASH was among
the best-performing methods of the ASH-variants, along with CATErr+ASH
and CATErr+MAD+ASH.

\begin{figure}
\begin{center}
\includegraphics[scale=0.9]{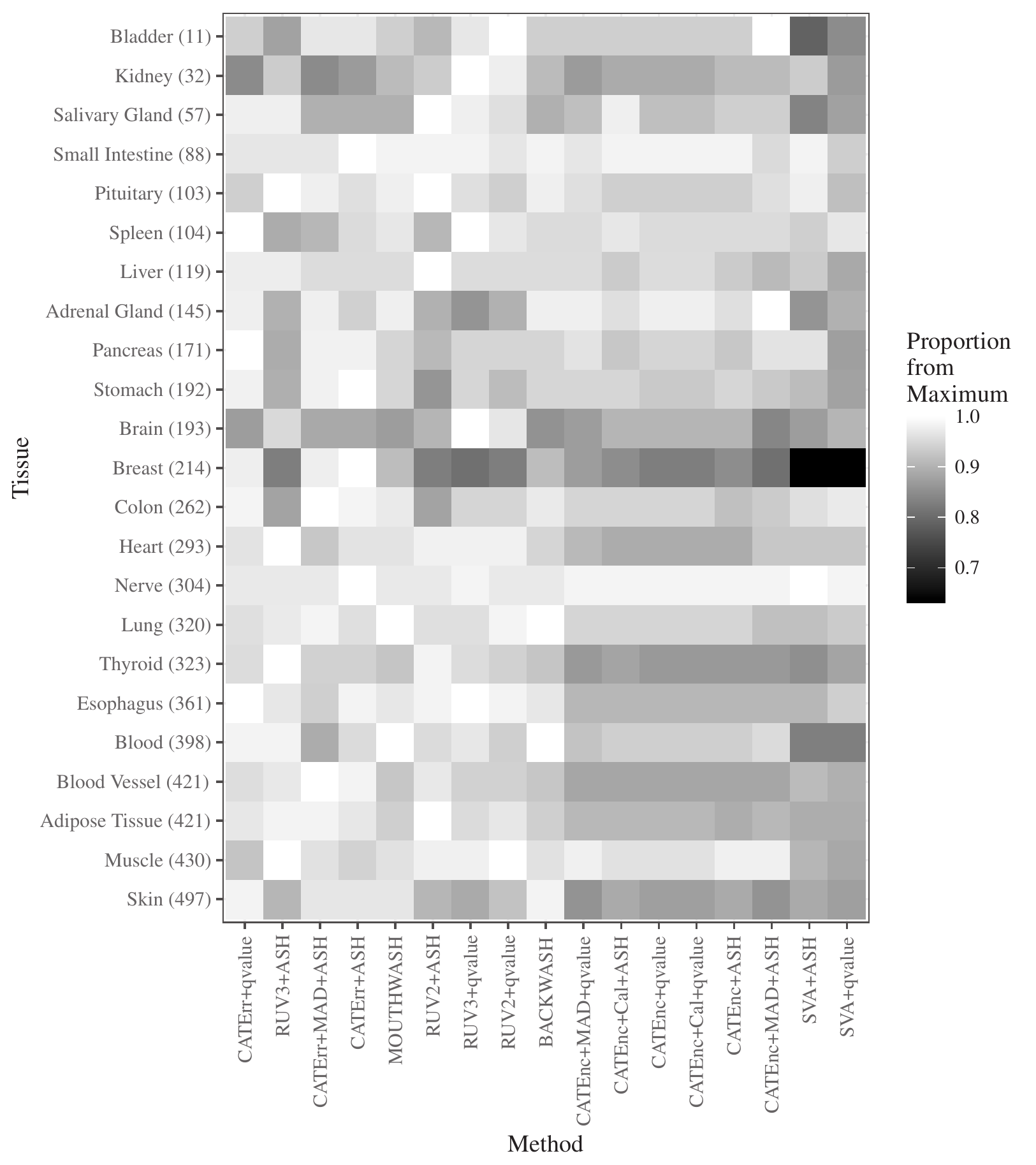}
\caption{Comparison of methods based on positive controls. For each
  method we computed the proportion $s$ of their most significant
  genes (testing for association with sex) that were on a sex
  chromosome. We then divided the $s$ for each method by the maximum
  $s$ among all methods. The tissues are ordered by sample size, and
  the methods are ordered by maximum median proportion. White
  indicates better performance than dark.}
\label{figure:prop.max}
\end{center}
\end{figure}

Though many methods performed similarly in ranking the most
significant genes, it would be wrong to think that they all produced
the same results. In particular, the methods differ considerably in
their assessments of significance and estimates of the proportion of
null genes ($\pi_0$). For example Table \ref{table:pi0hat} shows
median estimates of $\pi_0$ for each method across tissues. The
estimates range from 0.28 to almost 1. Generally \ash{}-based methods
produce smaller estimates of $\pi_0$ than qvalue-based methods, with
the exceptions of MOUTHWASH, BACKWASH, and those methods whose
variances were calibrated either using MAD or control genes. Though we
do not know the true value of $\pi_0$ here, and it is possible that
there are many non-sex chromosome genes with expression differences
between the sexes, it is interesting that MOUTHWASH and BACKWASH, the
best-performing methods in the simulations, estimate that most genes
are null.

\begin{table}[ht]
  \centering
  \caption{Median estimate of $\pi_0$ for each method across tissues
    when testing for differences between sexes.}
  \label{table:pi0hat}
  \begin{tabular}{lr}
    \hline
    Method            & $\hat{\pi}_0$ \\
    \hline
    SVA+ASH           & 0.28 \\
    CATErr+ASH        & 0.33 \\
    RUV3+ASH          & 0.38 \\
    OLS+ASH           & 0.40 \\
    RUV2+ASH          & 0.43 \\
    CATEnc+ASH        & 0.55 \\
    SVA+qvalue        & 0.70 \\
    RUV3+qvalue       & 0.75 \\
    CATErr+qvalue     & 0.76 \\
    CATEnc+qvalue     & 0.78 \\
    RUV2+qvalue       & 0.79 \\
    OLS+qvalue        & 0.80 \\
    CATEnc+Cal+ASH    & 0.89 \\
    CATEnc+Cal+qvalue & 0.90 \\
    CATErr+MAD+ASH    & 0.91 \\
    MOUTHWASH         & 0.99 \\
    CATEnc+MAD+ASH    & 0.99 \\
    BACKWASH          & 0.99 \\
    CATEnc+MAD+qvalue & 1.00 \\
    CATErr+MAD+qvalue & 1.00 \\
    \hline
  \end{tabular}
\end{table}

Another, perhaps still more striking, feature of the MOUTHWASH and
BACKWASH results is shown in Figure \ref{figure:rank.lfdr} which shows
the median lfdr for each \ash{}-based method as one moves down their
list of the top 500 most significant genes. For both MOUTHWASH and
BACKWASH the estimated lfdrs sharply increase from 0 at around 50-100
genes. Furthermore, this sharp increase occurs just where the ranking
starts to move away from genes on sex chromosomes (the shade moving
from black, red in the online version, to light grey).  Again, we do
not know the truth here, but the behavior of MOUTHWASH/BACKWASH is
consistent with most of the true differences being at genes on sex
chromosomes, and is strikingly different from most other methods. The
MAD-calibrated methods also exhibit this behavior. However, in
simulations with large sample sizes the MAD methods always estimated
few genes to be significant, even when half of the genes were
differentially expressed (Supplementary Figure \ref{figure:pi0.5}),
making it difficult to rely on their results.  The increase in lfdr of
CATEnc+Cal+ASH is not nearly as fast as that of MOUTHWASH and BACKWASH
and much less consistent across tissues (Supplementary Figure
\ref{figure:lfdr.full}).

\begin{figure}
\begin{center}
\includegraphics[scale = 0.9]{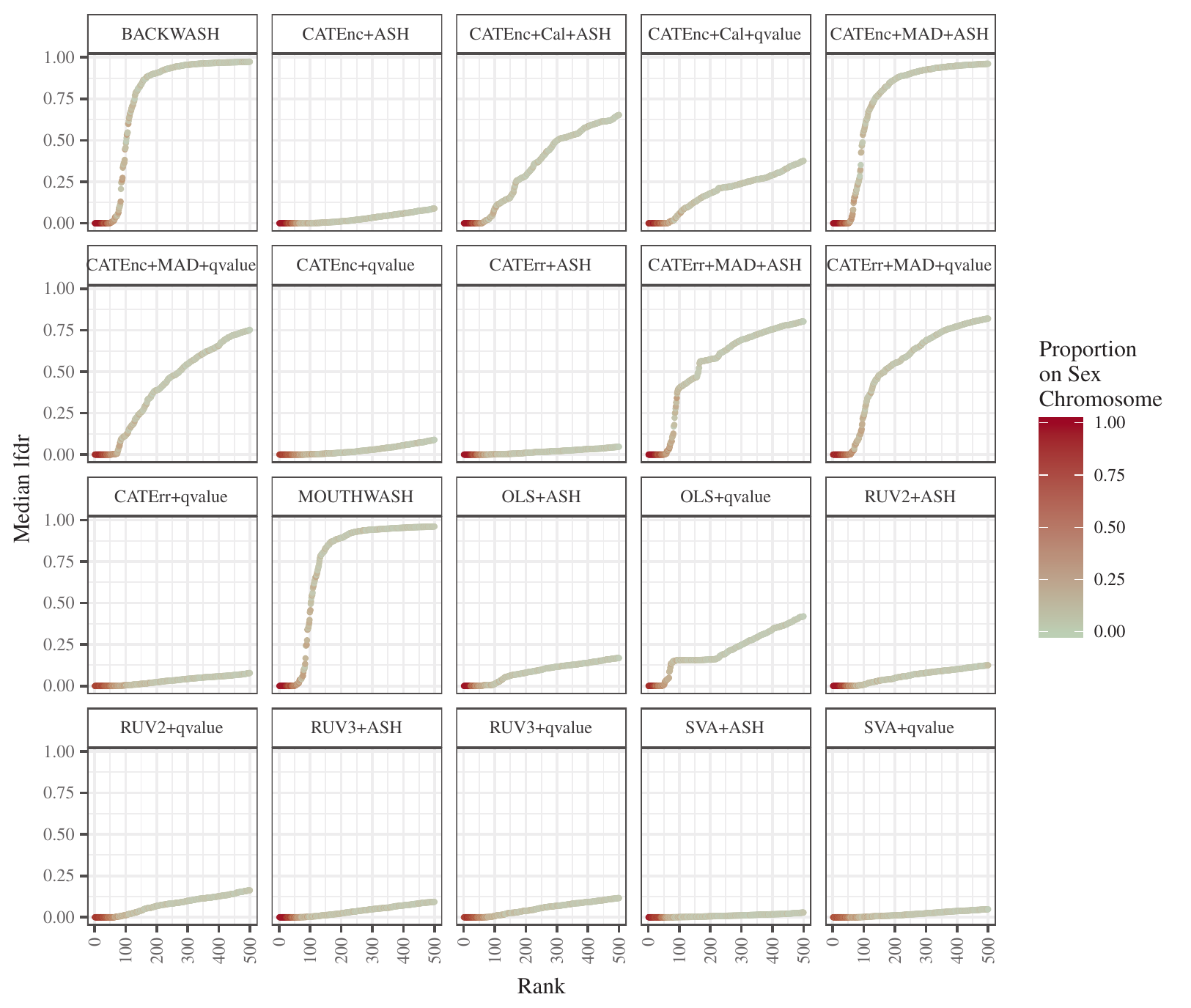}
\caption{Figure showing how median lfdr changes through the list of
  500 most significant genes.  For each method we sorted the lfdr's
  across genes in each tissue, and took the median lfdr across tissues
  at each rank. (Results for each tissue are in Supplementary Figure
  \ref{figure:lfdr.full}).  The color of each point indicates the
  proportion of tissues that have a sex chromosome gene at that rank
  (red indicating higher proportion).}
\label{figure:rank.lfdr}
\end{center}
\end{figure}

\section{Discussion}
\label{section:discussion}

We have presented a simple modular approach to combining
two key ideas for the analysis of genomic experiments:~EB shrinkage to
induce sparsity on effects, and FA to capture unwanted variation. Our
results demonstrate that these new methods have competitive
performance compared with a range of existing methods. They also
highlight that even when methods agree closely in their rankings of
genes (by strength of evidence against the null), they can vary widely
in their assessments of significance (e.g.~estimated FDRs).  Indeed,
even within a single ``method'', significance assessments can be
sensitive to details of how it is applied.  For example, in our
experience the way that variance estimates are dealt with can have a
very dramatic effect on estimated FDRs and related quantities. In
MOUTHWASH, the introduction of the variance inflation parameter $\xi$
has a substantial impact, and reduces the potential for
anti-conservative (under-)estimates of FDR.

Although we have used the term ``genomic experiments'', our methods
are really aimed at a particular type of genomic experiment:~where
there is a single covariate which may be associated with many measured
variables (e.g.~a differential expression experiment, where treatment
may affect the expression of many genes).  One different type of
genomic experiment that we do not address here is experiments to
identify ``expression Quantitative Trait Loci'' (eQTLs), which are
genetic variants associated with gene expression. The issues of sparse
effects, and unwanted variation, certainly arise when attempting to
identify eQTLs.  And some methods to deal with these issues have been
developed with a particular focus on eQTL studies
\citep[e.g.][]{stegle2008accounting, stegle2010bayesian,
  stegle2012using, fusi2012joint}. However, eQTL studies also differ
in a crucial way from the studies considered here. Specifically,
typical (population-based) eQTL studies involve many covariates
(different genetic variants), each of which is typically associated
with just one or a few genes (the strongest eQTLs are locally acting),
rather than a single covariate associated with many genes.  This
difference is fundamental:~when dealing with a single covariate that
may affect many genes, it is both particularly important and
particularly delicate to remove unwanted variation {\it without also
  removing the effect of interest}, whereas this issue is less
pressing in eQTL studies. (Population substructure in the genotype
data is a separate issue, which we do not discuss here.)  Indeed, in
eQTL studies, i) unwanted variation in expression data is rarely
associated with the covariates of interest, and so usually decreases
power rather than creating false positives; ii) when removing unwanted
variation one need not be too concerned about accidentally removing
signal of interest, and even very simple approaches such as using PCA
on the expression matrix typically improve power
\citep{pickrell2010understanding}. Neither of these hold in the
settings we focused on here.

One key feature of our approach is that, like many of the most popular
current approaches, it is designed to be {\it modular}. In particular,
although our results here are all based on using a simple FA
(truncated PCA), our methods could easily accommodate other approaches
to FA.  For example, it could accommodate Bayesian methods such as SFA
\citep{engelhardt2010analysis}, or the FA implemented in the software
PEER, which use a normal prior distribution on both the factors and
loadings \citep{stegle2010bayesian}. In principle there could be
objections to simply plugging these FAs into our approach:~for
example, the argument that the factor estimates are only identified up
to their row-space does not always hold for Bayesian FA, so the
property of MOUTHWASH that it depends on factor estimates only through
their row-space might be considered suspect. Put another way, one
could argue that when using prior distributions on the factors the
modular approach to fitting \eqref{eq:full.model} is suboptimal, and
could be improved by a purpose-built joint fitting routine. However,
the benefits of modular approaches are so great that it nonetheless
seems worthwhile to explore these ideas.

\section*{Software}
All methods introduced in this paper are implemented in the R package
\texttt{vicar} available at
\url{https://github.com/dcgerard/vicar}. Code to reproduce all results
in this paper is available at
\url{https://github.com/dcgerard/mouthwash_sims} (DOI:~10.5281/zenodo.1248856).

\section*{Funding}
This work was supported by the National Institutes of Health [grant
number HG002585]; and the Gordon and Betty Moore Foundation [Grant
number GBMF \#4559].

\section*{Acknowledgments}

Some of the original code for simulating the datasets in Section
\ref{section:simulations} was written by Mengyin Lu, to whom we give
our thanks.\\

\clearpage
\appendix
\section{Supplementary Materials}

\subsection{Simple illustration}
\label{section:simple.illustration}

We present a simple example that illustrates the need to address
unwanted variation.  We took the top 10,000 expressed genes of an
RNA-seq data on human muscle samples \citep{gtex2015} and randomly
sampled six individuals, which we randomly split into two groups. All
genes are thus theoretically ``null'' (unassociated with group
assignment). However, when we apply the \ash{} method from
\citet{stephens2016false} to the OLS estimates of $\bhat,\shat$ from
these null data, it infers that almost all genes are non-null
(estimated proportion of null genes, $\pi_0$, 0.0077), and indicates
almost every gene is significant with lfdr close to 0 (Supplementary
Figure \ref{figure:lfdr}, left panel). This behavior is not atypical
(Supplementary Figure \ref{figure:pi0.1}, top panels). Applying \ash{}
with effects estimated using a more sophisticated RNA-seq data
analysis pipeline instead of OLS \citep{law2014voom} slightly improved
matters (Supplementary Figure \ref{figure:lfdr}, second panel from the
left). In contrast applying MOUTHWASH and BACKWASH produced
essentially no significant genes, with lfdrs clustering closer to 1
(Supplementary Figure \ref{figure:lfdr}, right panels).

\subsection{Details of modular approach to fitting Factor-augmented Regression Model}
\label{section:detailed.review}

Many methods (e.g.~RUV4, LEAPP, and CATE) use a two-step approach to
fitting the factor-augmented regression model \eqref{eq:full.model}.
\citet{wang2017confounder} provide an elegant framing of this two-step
approach as a rotation followed by estimation in two independent
models.  Since this plays a key role in our methods we review it here.

For convenience we repeat the factor-augmented regression model here:
\begin{align}
  \label{eq:full.model2}
  \bs{Y}_{n\times p} = \bs{X}_{n \times k}\bs{\beta}_{k \times p} +
  \bs{Z}_{n \times q}\bs{\alpha}_{q \times p} + \bs{E}_{n\times p},
\end{align}
where we assume the number of samples $n$ is larger than the number of
covariates $k$.  As mentioned in the main text, we assume that only
one covariate is of interest. Without loss of generality, we will
assume that the ``uninteresting'' covariates are located in the first
$k - 1$ columns of $\bs{X}$ and the ``interesting'' covariate is in
the last column of $\bs{X}$. Thus we can partition
$\bs{\beta} =
\genfrac{(}{)}{0pt}{1}{\bs{\beta}_1}{\bs{\beta}_2^\intercal}$ so that
$\bs{\beta}_1 \in \mathbb{R}^{(k - 1) \times p}$ contains the
coefficients for the first $k - 1$ covariates and
$\bs{\beta}_2 \in \mathbb{R}^{p}$ contains the coefficients for the
covariate of interest.

Let $\bs{X} = \bs{Q}\bs{R}$ be the QR decomposition of $\bs{X}$, where
$\bs{Q} \in \mathbb{R}^{n \times n}$ is an orthogonal matrix
($\bs{Q}^{\intercal}\bs{Q} = \bs{Q}\bs{Q}^{\intercal} = \bs{I}_n$) and
$\bs{R}_{n \times k} = \genfrac{(}{)}{0pt}{1}{\bs{R}_1}{\bs{0}}$,
where $\bs{R}_1 \in \mathbb{R}^{k \times k}$ is an upper-triangular
matrix. Pre-multiplying \eqref{eq:full.model2} by $\bs{Q}^{\intercal}$
on both sides yields:~\def\tY{\tilde{\bs{Y}}} \def\tZ{\tilde{\bs{Z}}}
\def\tE{\tilde{\bs{E}}}
\begin{equation}
  \bs{Q}^{\intercal}\bs{Y} =
  \bs{R}\bs{\beta} +
  \bs{Q}^{\intercal}\bs{Z}\bs{\alpha} +
  \bs{Q}^{\intercal}\bs{E},
  \end{equation}
  which we write
  \begin{equation} \label{eq:rotated}
  \tY =
  \bs{R}\bs{\beta} +
  \tZ\bs{\alpha} +
  \tE
\end{equation}
where $\tY:=\bs{Q}^{\intercal}\bs{Y}$, $\tZ:=\bs{Q}^{\intercal}\bs{Z}$,$\tE:=\bs{Q}^{\intercal}\bs{E}$.

By exploiting the fact that $\bs{R}_1$ is upper triangular,
\eqref{eq:rotated} can be rewritten as:
\begin{align}
  \label{eq:Y1.model}
  \tY_1 &= \bs{R}_{11}\bs{\beta}_1 + \bs{r}_{12}\bs{\beta}_2^{\intercal} +  \tZ_1\bs{\alpha} +  \tE_1 \\
  \label{eq:Y2.model}
  \tilde{\bs{y}}_2^\intercal &= \phantom{\bs{R}_{11}\bs{\beta}_1 +\ }
   r_{22}\bs{\beta}_2^\intercal +
   \tilde{\bs{z}}_2^\intercal \bs{\alpha} +  \tilde{\bs{e}}_2^\intercal \\
  \label{eq:Y3.model}
  \tY_3 &=
 \phantom{\bs{R}_{11}\bs{\beta}_1 + \bs{r}_{12}\bs{\beta}_2^{\intercal} +\ }
   \tZ_3\bs{\alpha} +  \tE_3.
\end{align}
Here
\begin{equation}
  \bs{R}_1 =
  \left(
    \begin{array}{cc}
      \bs{R}_{11} & \bs{r}_{12}\\
      \bs{0} & r_{22}
    \end{array}
  \right),
\end{equation}
and we have conformably partitioned each of $\tY,\tZ,\tE$ into i)
their first $k-1$ rows; ii) their $k$th row; iii) the remaining $n-k$
rows, with for example
\begin{equation}
  \tY =
  \left(
    \begin{array}{c}
      \tY_1 \\
      \tilde{\bs{y}}_2^\intercal \\
      \tY_3
    \end{array}
  \right).
\end{equation}
We have used lower-case
$\tilde{\bs{y}}_2,\tilde{\bs{z}}_2,\tilde{\bs{e}}_2$ to indicate that
these quantities are vectors.

The error terms in \eqref{eq:Y1.model},
\eqref{eq:Y2.model}, and \eqref{eq:Y3.model} are independent, because
$\tE$ is equal in distribution to $\bs{E}$, which is matrix normal
\citep{srivastava1979introduction, dawid1981some} with independent
rows.

This rewriting suggests the following two-step estimation procedure,
which in essence is the approach used by RUV4, LEAPP, and CATE:
\begin{enumerate}
\item Estimate $\bs{\alpha}$ and the $\sigma_j$'s using
  \eqref{eq:Y3.model}, specifically by applying some kind of FA to $\tY_3$.
  Call these estimates $\hat{\bs{\alpha}}$ and
  $\hat{\sigma}_j$.
\item Estimate $\bs{\beta}_2$ and $\tilde{\bs{z}}_2$ given
  $\bs{\alpha}$ and the $\sigma_j$'s using \eqref{eq:Y2.model}, which
  can be written:
\begin{align}
\label{eq:reduced.model}
\tilde{\bs{y}}_2 \sim N_p(r_{22}\bs{\beta}_2 + \hat{\bs{\alpha}}^\intercal \tilde{\bs{z}}_2, \hat{\bs{\Sigma}}).
\end{align}
\end{enumerate}
As equation \eqref{eq:Y1.model} contains the nuisance parameters
$\bs{\beta}_1$, it is ignored.

In the main text we simplified the description by describing Step 1 as
applying FA to the matrix of residuals obtained from regressing the
columns of $\bs{Y}$ on $\bs{X}$ \eqref{eqn:FA}. As noted by
\citet{wang2017confounder}, for many choices of FA, applying FA to
$\tY_3$ is equivalent to applying FA to the residuals because $\tY_3$
and the matrix of residuals have the same sample covariance
matrix. (However, the mathematical derivation is clearer using
$\tY_3$, and our software implementation actually uses $\tY_3$.)

Both MOUTHWASH and BACKWASH use this approach.  Indeed, model
\eqref{eq:simplified.model} is the same as \eqref{eq:reduced.model}
with a simple change of notation:
\begin{align}
\label{eq:divide.r22} \hat{\bs\beta} &:= \bs{y}_2 / r_{22},\ \hat{\bs{\alpha}} := \hat{\bs{\alpha}} / r_{22},\ \bs{S} := \hat{\bs{\Sigma}} / r_{22}^2, \text{ and}\\
\label{eq:drop.2}\bs{z} &:= \tilde{\bs{z}}_2 \text{ and } \bs{\beta} := \bs{\beta}_2.
\end{align}
It is easy to show that $\hat{\bs\beta} = \bs{y}_2 / r_{22}$ are equal
to the OLS estimates of $\bs{\beta}_2$ obtained by regressing each
column of $\bs{Y}$ on $\bs{X}$.

\subsection{Estimating linear combinations of the rows of $\bs{\beta}$}
\label{section:linear.combo}

Suppose a researcher is interested not in a single row of
$\bs{\beta}$, but rather a single linear combination of $\bs{\beta}$,
$\bs{c}^{\intercal}\bs{\beta}$, for some $\bs{c} \in
\mathbb{R}^k$. For example, if one were interested in a simple
comparison of the effect of the first and second covariates,
$\beta_{1j} - \beta_{2j}$ (for all $j = 1,\ldots,p$), then
$\bs{c}^{\intercal} = (1, -1, 0, 0, \ldots, 0)$. As long as only one
linear combination of the rows of $\bs{\beta}$ is of interest,
MOUTHWASH and BACKWASH may be applied.

To do so, let the columns of $\bs{L} \in \mathbb{R}^{k \times (k-1)}$
be any orthonormal basis of the orthogonal complement of the space
spanned by $\bs{c}$ (e.g.~take the columns of $\bs{L}$ to be the first
$k-1$ eigenvectors of
$\bs{I}_{k-1} - \bs{c}\bs{c}^{\intercal}/\|\bs{c}\|^2$). Then,
assuming model \eqref{eq:full.model2}, we have
\begin{align}
  \label{eq:int.c.l.model}
  \bs{Y} = \bs{X}(\bs{c}/\|\bs{c}\|^2, \bs{L}) \genfrac{(}{)}{0pt}{1}{\bs{c}^{\intercal}}{\bs{L}^{\intercal}}\bs{\beta} +
  \bs{Z}\bs{\alpha} + \bs{E},
\end{align}
since
\begin{align}
  \genfrac{(}{)}{0pt}{1}{\bs{c}^{\intercal}}{\bs{L}^{\intercal}}^{-1} = (\bs{c}/\|\bs{c}\|^2, \bs{L}).
\end{align}
Now let $\tilde{\bs{X}} := (\bs{X}\bs{c}/\|\bs{c}\|^2, \bs{X}\bs{L})$
and
$\tilde{\bs{\beta}} := (\bs{\beta}^{\intercal}\bs{c},
\bs{\beta}^{\intercal}\bs{L})^{\intercal}$. Then equation
\eqref{eq:int.c.l.model} is equal to
\begin{align}
  \bs{Y} = \tilde{\bs{X}}\tilde{\bs{\beta}} +
  \bs{Z}\bs{\alpha} + \bs{E},
\end{align}
where the first row of $\tilde{\bs{\beta}}$ is equal to
$\bs{c}^{\intercal}\bs{\beta}$. We may now apply the modular approach
used to fit MOUTHWASH and BACKWASH (as in Section
\ref{section:detailed.review}) using $\tilde{\bs{X}}$ instead of
$\bs{X}$. Here, the first column of $\bs{\tilde{X}}$ is the covariate
of interest and its corresponding coefficients (the first row of
$\tilde{\bs{\beta}}$) represent the linear combination of the rows of
$\bs{\beta}$ that are of interest.

\subsection{MOUTHWASH optimization details}
\label{section:em}

\subsubsection{EM algorithm for normal likelihood and normal mixtures}
\label{section:em.normal}

Here we describe the EM algorithm used for solving the optimization
step \eqref{eq:mouth.opt} in MOUTHWASH when the mixture components in
\eqref{eq:prior.mixture} are normal. (For the generalization to a
$t_\nu$ likelihood and the case where the mixture components are
uniform see the coordinate ascent updates in the next subsection).

The model is:
\begin{align}
p(\hat{\vbeta}|\bs{z},\bs{\beta}, \xi) &= \prod_{j = 1}^pN(\hat{\beta}_j|\beta_j + \hat{\bs{\alpha}}_j^{\intercal}\bs{z}, \xi s_{jj}^2)\\
p(\bs{\beta}) &= \prod_{j = 1}^p g(\beta_j|\bs{\pi})\\
g(\beta_j|\bs{\pi}) &= \pi_0\delta_0(\beta_j) + \sum_{m = 1}^M\pi_mN(\beta_j|0,\tau_m^2). \label{eq:gnormalmix}
\end{align}
By integrating over $\bs{\beta}$, we have
\begin{align}
p(\hat{\vbeta}|\bs{z},\bs{\pi}, \xi) &= \prod_{j=1}^pp(\hat{\beta}_j|\bs{z},\bs{\pi}, \xi) \label{eq:Y.integrateB}\\
p(\hat{\beta}_j|\bs{z},\bs{\pi}, \xi) &= \pi_0N(\hat{\beta}_j|\hat{\bs{\alpha}}_j^{\intercal}\bs{z},\xi s_{jj}^2) + \sum_{m=1}^M\pi_m N(\hat{\beta}_j|\hat{\bs{\alpha}}_j^{\intercal}\bs{z},\xi s_{jj}^2 + \tau_m^2).\label{eq:Yj.integrateB}
\end{align}

Our goal is to maximize the likelihood \eqref{eq:Y.integrateB}
over $\bs{\pi}$, $\bs{z}$, and $\xi$. In fact
we consider the slightly more general problem
of optimizing the penalized
likelihood
\begin{align}
p(\hat{\vbeta}|\bs{z},\bs{\pi}, \xi)h(\bs{\pi}|\bs{\lambda}),
\end{align}
where $h(\bs{\pi}|\bs{\lambda})$ is defined in \eqref{eq:penalty.lambda}.

To develop the EM algorithm, we use the usual
approach for mixtures, introducing indicator variables
that indicate which component of the mixture
(\ref{eq:gnormalmix}) gave rise to each $\beta_j$.
Let $\bs{w}_j = (w_{0j},\ldots,w_{Mj})^{\intercal}$ denote a one-of-$(M + 1)$
indicator vector representing the mixture component
that gave rise to $\beta_j$, so
$\sum_{m=0}^M w_{mj}= 1$ and $p(w_{mj}=1) = \pi_m$. Then the complete
data likelihood is:
\begin{align}
\begin{split}
&p(\hat{\vbeta},\bs{W}|\bs{z},\bs{\pi}, \xi)h(\bs{\pi}|\bs{\lambda}) \\
=& \left(\prod_{m=0}^M\pi_m^{\lambda_m - 1}\right)\prod_{j=1}^p\exp\left\{\sum_{m=0}^Mw_{mj}\log(\pi_m)-\left(\sum_{m=0}^M\frac{w_{mj}}{2(\xi s_{jj}^2 + \tau_m^2)}\right)(\hat{\beta}_j - \hat{\bs{\alpha}}_j^{\intercal}\bs{z})^2\right.\\
&\left.-\frac{1}{2}\sum_{m=0}^Mw_{mj}\log(\xi s_{jj}^2 + \tau_m^2) - \frac{1}{2}\log(2\pi)\right\}.
\end{split}
\end{align}
And the complete data log-likelihood is:
\begin{align}
\begin{split}
l_\text{complete}(\bs{z},\bs{\pi},\xi; \hat{\bs\beta},\bs{W}) &:=\sum_{j=1}^p\left\{\sum_{m=0}^Mw_{mj}\log(\pi_m) -\left(\sum_{m=0}^M\frac{w_{mj}}{2(\xi s_{jj}^2 + \tau_m^2)}\right)(\hat{\beta}_j - \hat{\bs{\alpha}}_j^{\intercal}\bs{z})^2 -\right.\\
&\left.\frac{1}{2}\sum_{m=0}^Mw_{mj}\log(\xi s_{jj}^2 + \tau_m^2) - \frac{1}{2}\log(2\pi)\right\} + \sum_{m = 0}^M(\lambda_m - 1)\log(\pi_m). \label{eq:complete}
\end{split}
\end{align}

Let $\bs{\pi}^{(old)}$, $\bs{z}^{(old)}$, and $\xi^{(old)}$ be the
current values of the parameters. Then
\begin{align}
p(w_{mj}=1|\hat{\beta}_j,\bs{z}^{(old)},\bs{\pi}^{(old)}, \xi^{(old)}) = \frac{\pi_m^{(old)}N(\hat{\beta}_j|\hat{\bs{\alpha}}_j^{\intercal}\bs{z}^{(old)},\xi^{(old)} s_{jj}^2 + \tau_m^2)}{\sum_{i = 0}^M\pi_i^{(old)}N(\hat{\beta}_j|\hat{\bs{\alpha}}_j^{\intercal}\bs{z}^{(old)},\xi^{(old)} s_{jj}^2 + \tau_i^2)} =:~q_{mj}. \label{eq:def.tmj}
\end{align}

The E-step of the EM algorithm involves forming
the expected complete data log-likelihood,
which simply involves replacing $w_{kj}$ with $q_{kj}$
in \eqref{eq:complete}:
\begin{align}
\begin{split}
&\sum_{j=1}^p\left\{\sum_{m=0}^Mq_{mj}\log(\pi_m) -\left(\sum_{m=0}^M\frac{q_{mj}}{2(\xi s_{jj}^2 + \tau_m^2)}\right)(\hat{\beta}_j - \hat{\bs{\alpha}}_j^{\intercal}\bs{z})^2 -\right.\\
&\left.\frac{1}{2}\sum_{m=0}^Mq_{mj}\log(\xi s_{jj}^2 + \tau_m^2) - \frac{1}{2}\log(2\pi)\right\} + \sum_{m = 0}^M(\lambda_m - 1)\log(\pi_m). \label{eq:given.t}
\end{split}
\end{align}
The M-step then involves optimizing this over $\bs{z}$, $\bs{\pi}$, and $\xi$.

The update for $\pi$ follows by recognizing the kernel of a
multinomial likelihood
  \begin{align}
    \pi_m &\leftarrow \frac{\sum_{j=1}^pq_{mj} + \lambda_m - 1}{\sum_{\ell=0}^M\left(\sum_{j = 1}^pq_{\ell j} + \lambda_{\ell} - 1\right)}\\
    &= \frac{\sum_{j=1}^pq_{mj} + \lambda_m - 1}{\sum_{\ell=0}^M\sum_{j = 1}^pq_{\ell j} + \sum_{\ell=0}^M\lambda_{\ell} - M}\\
    &= \frac{\sum_{j=1}^pq_{mj} + \lambda_m - 1}{p - M + \sum_{\ell=0}^M\lambda_{\ell} }.
  \end{align}
In the case when there is no penalty, $\lambda_1 = \cdots = \lambda_M = 1$, we have
\begin{align}
\pi_m \leftarrow \frac{1}{p}\sum_{j=1}^pq_{mj}.
\end{align}

We then perform a few iterative updates on $\xi$ and $\bs{z}$. To
update $\bs{z}$ given $\xi$, we note that optimizing
\eqref{eq:given.t} over $\bs{z}$ is the same as weighted linear
regression with diagonal weight (precision) matrix
$\bs{\Theta}_{\xi} \in \mathbb{R}^{p\times p}$ with diagonal elements
$\theta_{\xi, jj} = \sum_{m=0}^M\frac{q_{mj}}{\xi s_{jj}^2 +
  \tau_m^2}$. We get
\begin{align}
\bs{z} \leftarrow (\hat{\bs{\alpha}}\bs{\Theta}_{\xi}\hat{\bs{\alpha}}^{\intercal})^{-1}\hat{\bs{\alpha}}\bs{\Theta}_{\xi} \hat{\vbeta}.
\end{align}
To update $\xi$ given $\bs{z}$ we can use some standard univariate
optimizer, such as Brent's method \citep{brent1971algorithm}.

One step of this EM algorithm is presented in Algorithm
\ref{algorithm:normal.normal}. Iteratively performing the steps in
Algorithm \ref{algorithm:normal.normal} is guaranteed to increase the
likelihood toward a local maximum.

\begin{algorithm}
  \caption{EM Algorithm for Normal Mixtures Prior and Normal Likelihood}
  \label{algorithm:normal.normal}
  \begin{algorithmic}[1]
    \STATE Given the current values of the parameters in our model, $\bs{\pi}^{(old)}$, $\bs{z}^{(old)}$,
    and $\xi^{(old)}$, let $q_{mj}$ be defined as in \eqref{eq:def.tmj}.
    \STATE Set $\pi_m = \frac{\sum_{j=1}^pq_{mj} + \lambda_m - 1}{p - M + \sum_{\ell=0}^M\lambda_{\ell} }$,
    \REPEAT
    \STATE Let $\bs{\Theta}_{\xi}$ be a diagonal matrix with diagonal elements $\theta_{\xi, jj} = \sum_{m=0}^M\frac{q_{mj}}{\xi s_{jj}^2 + \tau_m^2}$.
    \STATE Set $\bs{z} = (\hat{\bs{\alpha}}\bs{\Theta}_{\xi}\hat{\bs{\alpha}}^{\intercal})^{-1}\hat{\bs{\alpha}}\bs{\Theta}_{\xi} \hat{\vbeta}$.
    \STATE Update $\xi$ given $\bs{z}$ and $\bs{\pi}$ by maximizing \eqref{eq:given.t} using Brent's method.
    \UNTIL {convergence}
  \end{algorithmic}
\end{algorithm}

\subsubsection{Coordinate Ascent for $t_{\nu}$-Uniform Problem}
\label{section:em.uniform}

Here we describe the optimization steps used for the generalization of
MOUTHWASH to a $t_\nu$ likelihood \eqref{eq:t.likelihood} in the case
where the mixture components are uniform. (This also applies to the
normal likelihood with uniform components by setting $\nu=\infty$).

The model is:
\begin{align}
p(\hat{\vbeta}|\bs{z},\bs{\beta}, \xi) &= \prod_{j = 1}^pt_{\nu}(\hat{\beta}_j|\beta_j + \hat{\bs{\alpha}}_j^{\intercal}\bs{z}, \xi s_{jj}^2)\\
p(\bs{\beta}) &= \prod_{j = 1}^p g(\beta_j|\bs{\pi})\\
g(\beta_j|\bs{\pi}) &= \pi_0\delta_0(\beta_j) + \sum_{m = 1}^M\pi_mU(\beta_j|a_m,b_m).
\end{align}
By integrating over $\bs{\beta}$, we have
\begin{align}
p(\hat{\vbeta}|\bs{z},\bs{\pi}, \xi) &= \prod_{j=1}^pp(\hat{\beta}_j|\bs{z},\bs{\pi}, \xi) \label{eq:Y.integrateB.t}\\
p(\hat{\beta}_j|\bs{z},\bs{\pi}, \xi) &= \pi_0t_{\nu}(\hat{\beta}_j|\hat{\bs{\alpha}}_j^{\intercal}\bs{z},\xi s_{jj}^2) + \sum_{m=1}^M\pi_m \tilde{f}_m(\hat{\beta}_j|\bs{z},\xi),\label{eq:Yj.integrateB.t}
\end{align}
where
\begin{align}
\tilde{f}_m(\hat{\beta}_j|\bs{z},\xi) =
\frac{T_{\nu}((\hat{\beta}_j - \hat{\bs{\alpha}}_j^{\intercal}\bs{z} - a_m)/(\xi^{1/2} s_{jj})) - T_{\nu}((\hat{\beta}_j - \hat{\bs{\alpha}}_j^{\intercal}\bs{z}- b_m)/(\xi^{1/2} s_{jj}))}{b_m - a_m}
\end{align}
where $T_{\nu}$ is the cdf of a standard $t_{\nu}$ distribution. For ease of notation, we will also let $\tilde{f}_0(\hat{\beta}_j|\bs{z}, \xi) := t_{\nu}(\hat{\beta}_j|\hat{\bs{\alpha}}_j^{\intercal}\bs{z},\xi s_{jj}^2)$.

To maximize the marginal likelihood \eqref{eq:Y.integrateB.t}, or
rather the log-likelihood,
\begin{align}
\label{eq:t.unif.ll}
\sum_{j = 1}^p\log p(\hat{\beta}_j|\bs{z}, \bs{\pi}, \xi),
\end{align}
we implemented a coordinate ascent algorithm to iteratively update
$\bs{z}$, $\bs{\pi}$, and $\xi$. To update $\bs{\pi}$ conditional on
$\bs{z}$ and $\xi$, we apply the same convex optimization procedure
described in \citet{stephens2016false} using the \texttt{ashr} package
\citep{stephens2016ashr}. To update $\xi$ given $\bs{\pi}$ and
$\bs{z}$, we use a standard univariate optimizer,
Brent's method \citep{brent1971algorithm}.

To update $\bs{z}$ given
$\bs{\pi}$ and $\xi$, we calculated the gradient of
\eqref{eq:t.unif.ll} with respect to $\bs{z}$:
\begin{align}
\label{eq:t.unif.grad}
\sum_{j = 1}^p \hat{\bs{\alpha}}_j \frac{\sum_{m = 0}^M\pi_m\bar{f}_{m}(\hat{\beta}_j|\bs{z})}{p(\hat{\beta}_j|\bs{z}, \bs{\pi}, \xi)},
\end{align}
where
\begin{align}
\bar{f}_{0}(\hat{\beta}_j|\bs{z}) &= \frac{(\nu + 1)(\hat{\beta}_j - \hat{\bs{\alpha}}_j^{\intercal}\bs{z})}{\nu \xi s_{jj}^2 + (\hat{\beta}_j - \hat{\bs{\alpha}}_j^{\intercal}\bs{z})^2}t_{\nu}(\hat{\beta}_j|\hat{\bs{\alpha}}_j^{\intercal}\bs{z}, \xi s_{jj}^2), \text{ and}\\
\bar{f}_{m}(\hat{\beta}_j|\bs{z}) &= \left(\frac{1}{b_m - a_m}\right) \left(t_{\nu}(\hat{\beta}_j|\hat{\bs{\alpha}}_j^{\intercal}\bs{z} - b_m, \xi s_{jj}^2) - t_{\nu}(\hat{\beta}_j|\hat{\bs{\alpha}}_j^{\intercal}\bs{z} - a_m, \xi s_{jj}^2)\right).
\end{align}
We then use a quasi-Newton approach to maximize
\eqref{eq:t.unif.ll} over $\bs{z}$ using
\eqref{eq:t.unif.grad} (specifically we used the BFGS method).

\subsection{Identifiability}
\label{section:identifiability}
\begin{theorem}
\label{theorem:row.space}
For all non-singular $\bs{A} \in \mathbb{R}^{q \times q}$, we have
that
\begin{align*}
\hat{g} = \argmax_{g \in \mathcal{U}}\max_{\bs{z}\in \mathbb{R}^q} p(\hat{\vbeta}|g,\bs{z},\hat{\bs{\alpha}},\bs{S}) = \argmax_{g\in\mathcal{U}}\max_{\bs{z} \in \mathbb{R}^q}p(\hat{\vbeta}|g,\bs{z},\bs{A}\hat{\bs{\alpha}},\bs{S}).
\end{align*}
\end{theorem}
\begin{proof}
\begin{align}
\label{eq:swap.A}\argmax_{g\in\mathcal{U}}\max_{\bs{z} \in \mathbb{R}^q}p(\hat{\vbeta}|g,\bs{z},\bs{A}\hat{\bs{\alpha}},\bs{S}) &= \argmax_{g\in\mathcal{U}}\max_{\bs{z} \in \mathbb{R}^q}p(\hat{\vbeta}|g,\bs{A}^{\intercal}\bs{z},\hat{\bs{\alpha}},\bs{S})\\
\label{eq:opt.ZA}&= \argmax_{g\in\mathcal{U}}\max_{\bs{A}^{\intercal}\bs{z} \in \mathbb{R}^q}p(\hat{\vbeta}|g,\bs{A}^{\intercal}\bs{z},\hat{\bs{\alpha}},\bs{S})\\
\label{eq:relabel.AZ}&= \argmax_{g\in\mathcal{U}}\max_{\bs{z} \in \mathbb{R}^q}p(\hat{\vbeta}|g,\bs{z},\hat{\bs{\alpha}},\bs{S}),
\end{align}
where \eqref{eq:swap.A} follows because
$(\bs{A}\hat{\bs{\alpha}})^{\intercal}\bs{z} =
\hat{\bs{\alpha}}^{\intercal}(\bs{A}^{\intercal}\bs{z})$,
\eqref{eq:opt.ZA} follows because optimizing over $\bs{z}$ is the same
as optimizing over $\bs{A}^{\intercal}\bs{z}$ for any non-singular
$\bs{A}$, and \eqref{eq:relabel.AZ} follows from relabeling
$\bs{A}^{\intercal}\bs{z}$ to be $\bs{z}$.
\end{proof}

\subsection{Mouthwash, additional Bells and Whistles}
\label{section:additional.bells}

Here we describe additional features we have implemented in
MOUTHWASH (see also Section \ref{section:mouthwash}).

\subsubsection{Effects that depend on standard errors}

\citet{stephens2016false} modified \eqref{eq:beta.prior.spec} to allow
the $\beta_j$'s to depend on the standard errors of the
$\bhat_j$'s. This may make sense, for example, in gene expression
studies if genes with higher variability tend to have larger
effects. Specifically, \citet{stephens2016false} set
\begin{align}
\label{eq:prior.gamma}
\frac{\beta_j}{s_j^{\gamma}}|s_j \overset{iid}{\sim} g,
\end{align}
where $\gamma \geq 0$ is specified. Estimating $g$ under
\eqref{eq:prior.gamma} is straightforward except when both
$\gamma = 1$ and we include the variance inflation parameter $\xi$
from \eqref{eq:mouth.opt.var}. Under these conditions $g$ and $\xi$
become non-identifiable.

To see this, consider the simple case with no unwanted variation
($\bs{z}=\bs{0}$), and write the normal term from
\eqref{eq:mouth.opt.var} as
\begin{align}
\hat{\beta}_j / s_j \overset{d}{=} \beta_j / s_j + e_j, \text{ where } e_j \overset{iid}{\sim} N(0, \xi).
\end{align}
So effectively $\beta_j/s_j + e_j$ are now {\it iid} observations from
a convolution of a distribution $g$ that is unimodal at 0 with a
$N(0,\xi)$ distribution. This convolution is itself unimodal, and ---
whatever the true value of $\xi$ --- could be fit perfectly using
$\xi=0$ and $g$ equal to the true $g$ convolved with the true
$N(0,\xi)$. Thus it is impossible to guarantee accurate estimation of
$\xi$ without making additional assumptions.

Although it is impossible to guarantee accurate estimation of $\xi$,
it {\it is} possible to guarantee conservative (over-)estimates of
$\xi$. This is formalized in the following lemma:
\begin{lemma}
  \label{lemma:deconvolution}
  For any distribution function, say $F$, unimodal about 0, there
  exists a maximal $\xi$ such that $F$ can be deconvolved into a
  $N(0, \xi)$ distribution function and another distribution function
  $G$ that is \emph{also unimodal} about 0. That is, making $\xi$ any
  larger would result in a non-unimodal $G$.
\end{lemma}
See Appendix \ref{section:lemma.proof} for proof.

Over-estimating $\xi$ is conservative in that it will over-shrink
estimates of $\beta$ and over-estimate FDR. Motivated by Lemma 1 we
can achieve this conservative behavior by introducing a small penalty
term to encourage $\xi$ to be as big as possible.  Specifically we
maximize the penalized likelihood:
\begin{align}
p(\hat{\vbeta}|\bs{\beta}, \bs{z}, \bs{S}, \xi)f(\xi|\lambda_{\xi}).
\end{align}
where
\begin{align}
f(\xi|\lambda_{\xi}) = \exp\{-\lambda_{\xi} / \xi\},
\end{align}
and $\lambda_{\xi} > 0$ is a penalty parameter that can be (in
principle) arbitrarily small.  Because $f(\xi|\lambda_{\xi})$ is
increasing, the introduction of this term promotes $\xi$ to be as
large as possible with $g$ unimodal.

\subsubsection{Generalizing normal likelihood to $t$ likelihood}

For small sample sizes the normality assumption in
\eqref{eq:simplified.model} might be better replaced with a $t$
distribution:
\begin{align}
\label{eq:t.likelihood}
p(\hat{\vbeta}|\bs{\beta}, \bs{z}, \bs{S}) = \prod_{j = 1}^p t_{\nu}(\hat{\beta}_j|\beta_j + \hat{\bs{\alpha}}_j^{\intercal}\bs{z}, s_{jj}^2),
\end{align}
where $t_{\nu}(\cdot|a, b^2)$ denotes the density of a (generalized)
$t$-distribution with degrees
of freedom $\nu$, location parameter $a$, and scale parameter
$b > 0$. That is,
\begin{align}
\label{eq:t.density}
t_{\nu}(\hat{\beta} | a, b^2) = \frac{\Gamma\left(\frac{\nu + 1}{2}\right)}{\Gamma\left(\frac{\nu}{2}\right)\sqrt{\pi\nu b^2}}\left(1 + \frac{(\hat{\beta} - a)^2}{\nu b^2}\right)^{-\frac{\nu + 1}{2}},
\end{align}
where $\Gamma(\cdot)$ is the gamma function.  A similar generalization
was implemented in \citep{stephens2016false}. This replacement of a
normal likelihood with a $t$ does not greatly complicate computations
when the mixture components in \eqref{eq:prior.mixture} are uniform,
and we have implemented this case (Appendix
\ref{section:em.uniform}). The normal case is more complex and not
implemented.

\subsubsection{Penalty on $\pi_0$ to promote conservative behavior}

\citet{stephens2016false} included the option of incorporating a
penalty on the mixing proportions to promote conservative (over-)
estimation of $\pi_0$. We also include this option here. Specifically
we allow a penalty of the form
\begin{align}
  \label{eq:penalty.lambda}
  h(\bs{\pi}|\bs{\lambda}) = \prod_{m=0}^M\pi_m^{\lambda_m - 1},
\end{align}
and maximize the penalized likelihood
\begin{align}
p(\hat{\vbeta}|\bs{\beta}, \bs{z}, \bs{S})h(\bs{\pi}|\bs{\lambda}),
\end{align}
where $p(\hat{\vbeta}|\bs{\beta}, \bs{z},
\bs{S})$ is defined in either
\eqref{eq:simplified.model} or \eqref{eq:t.likelihood}.
We use the same default value for $\bs{\lambda}$ as
\citet{stephens2016false}:~$\lambda_0 = 10$ and
$\lambda_i = 1$ for $i = 1,\ldots, m$. This encourages
conservative (over-) estimation of $\pi_0$, which is often considered
desirable in FDR contexts.

\subsubsection{Reducing computation for large $p$}

MOUTHWASH is computationally practical for typical gene-expression
studies, where $p\approx$ 20,000 genes.  However, in contexts where
$p$ exceeds 100,000 \citep[e.g.\ ChIP-seq,][]{ward2018silencing} the
run time can become inconvenient.  To reduce run-time in such cases we
suggest estimating $\bs{z}$ from \eqref{eq:simplified.model} using a
random subset of variables. As $\bs{z}$ typically contains at most a
few dozen parameters, a modest-sized subset should provide reasonable
estimates.

Specifically, we implemented the following speed-up strategy for $p$
very large.  First estimate $g,\bs{z}$ using a random subset of
variables.  Second, fixing the estimate of $\bs{z}$ from the first
step, re-estimate $g$ by maximum likelihood over all $p$ variables
(which is a convex optimization problem that can be solved efficiently
even for very large $p$).

\subsection{BACKWASH}
\label{section:backwash}

MOUTHWASH maximizes over $\bs{z}$ in \eqref{eq:mouth.opt}.  We now
describe an alternative that aims to better allow for uncertainty in
$\bs{z}$ by placing a prior $p(\bs{z})$ on $\bs{z}$ and integrating
out $\bs{z}$ when optimizing over $g$. Because of the introduction of
a prior distribution on $\bs{z}$ we call this approach BACKWASH for
\textbf{B}ayesian \textbf{A}djustment for \textbf{C}onfounding
\textbf{K}nitted \textbf{W}ith \textbf{A}daptive \textbf{SH}rinkage.
Specifically, BACKWASH replaces Step 3 of MOUTHWASH with:
\begin{enumerate}[noitemsep, nolistsep]
\item[3.] Estimate $g$ by:
  \begin{align}
    \hat{g} &:= \argmax_{g \in \mathcal{U}} p(\bhat |g, \hat{\bs{\alpha}}, \bs{S})\\
&= \argmax_{g \in \mathcal{U}} \prod_{j = 1}^p\int_{\beta_j}\int_{\bs{z}}N(\bhat_j |\beta_j + \hat{\bs{\alpha}}^{\intercal}\bs{z}, s_{jj}^2)g(\beta_j)p(\bs{z})\dif\bs{z}\dif\beta_j.
  \end{align}
  \end{enumerate}

  To specify the prior $p(\bs{z})$, we require that inference depends
  on $\hat{\bs{\alpha}}$ only through its rowspace (see Section
  \ref{section:mouthwash}).  A prior that satisfies this
  requirement is the so-called ``$g$-prior''
  \citep{zellner1986assessing, liang2008mixtures}:
\begin{align}
\label{eq:zg.prior}
\bs{z} \, | \, \hat{\bs{\alpha}} \sim N_q(\bs{0}, \phi^2(\hat{\bs{\alpha}}\hat{\bs{\alpha}}^{\intercal})^{-1}),
\end{align}
where $\phi\in\mathbb{R}^{+}$ is a hyperparameter that we estimate by
maximum marginal likelihood. With this prior the marginal likelihood
is
\begin{align}
\label{eq:marginal.out.z}
\int_{\bs{\beta}}N_p(\vbhat |\bs{\beta}, \bs{S} + \phi^2\hat{\bs{\alpha}}^{\intercal}(\hat{\bs{\alpha}}\hat{\bs{\alpha}}^{\intercal})^{-1}\hat{\bs{\alpha}})\prod_{j = 1}^pg(\beta_j)\dif\beta_j,
\end{align}
which depends on $\hat{\bs{\alpha}}$ only through its rowspace.

When we include the estimation of the hyperparameter $\phi$, and a
variance scaling parameter $\xi\in\mathbb{R}^+$ (Section
\ref{section:var.inflate}) the full BACKWASH Step 3 becomes:
\begin{enumerate}[noitemsep, nolistsep]
\item[3.] Let
  \begin{align}
    (\hat{g}, \hat{\phi}, \hat{\xi}) &:= \argmax_{(g, \phi, \xi)\ \in\ \mathcal{U} \times \mathbb{R}^+ \times \mathbb{R}^+} p(\vbhat |g,\phi, \xi, \hat{\bs{\alpha}}, \bs{S})\\
\label{eq:xi.phi.g.max}&= \argmax_{(g, \phi, \xi)\ \in\ \mathcal{U} \times \mathbb{R}^+ \times \mathbb{R}^+}\int_{\bs{\beta}}N_p(\vbhat |\bs{\beta}, \xi\bs{S} + \phi^2\hat{\bs{\alpha}}^{\intercal}(\hat{\bs{\alpha}}\hat{\bs{\alpha}}^{\intercal})^{-1}\hat{\bs{\alpha}})\prod_{j = 1}^pg(\beta_j)\dif\beta_j.
  \end{align}
\end{enumerate}

Maximizing \eqref{eq:xi.phi.g.max} is difficult, and so we resort to a
variational approximation \citep{blei2017variational} and instead
maximize a lower bound for the marginal likelihood over $g$, $\phi$,
and $\xi$ (see Appendix \ref{section:vem.backwash} for details).

\subsection{Variational EM Algorithm for BACKWASH}
\label{section:vem.backwash}

In this section, we present the Variational Expectation Maximization
(VEM) algorithm that we developed for the BACKWASH procedure in
Section \ref{section:backwash}. For a good introduction to variational
methods, see \citet{bishop2006pattern}. The model in Section
\ref{section:backwash} is
\begin{align}
\label{eq:y.like.back}[\hat{\vbeta}|\bs{\beta}, \phi, \xi] &\sim N_p(\bs{\beta}, \xi\bs{S} + \phi^2\hat{\bs{\alpha}}^{\intercal}(\hat{\bs{\alpha}}\hat{\bs{\alpha}}^{\intercal})^{-1}\hat{\bs{\alpha}})\\
\label{eq:beta.prior.back}\beta_j &\text{ i.i.d. s.t. } p(\beta_j) = \sum_{m = 0}^M\pi_mN(\beta_j|0, \tau_m^2),
\end{align}
where the $\tau_m$'s are known. Let
\begin{align}
\bs{A} := \hat{\bs{\alpha}}^{\intercal}(\hat{\bs{\alpha}}\hat{\bs{\alpha}}^{\intercal})^{-1/2} \in \mathbb{R}^{p \times q}.
\end{align}
We augment model
\eqref{eq:y.like.back}-\eqref{eq:beta.prior.back} with a
standard Gaussian vector $\bs{v} \in \mathbb{R}^{q}$ and
1-of-$M$ binary vectors $\bs{w}_j \in \mathbb{R}^{M}$, $j =
1,\ldots, p$. Then
\eqref{eq:y.like.back}-\eqref{eq:beta.prior.back} may be
equivalently represented by
\begin{align}
\hat{\vbeta} &\overset{d}{=} \bs{\beta} + \phi\bs{A}\bs{v} + \bs{e}\\
\bs{v} &\sim N_q(\bs{0}, \bs{I}_q)\\
\bs{e} &\sim N_p(\bs{0}, \xi\bs{S})\\
p(\beta_j, \bs{w}_j) &= \prod_{m = 0}^M\left[\pi_mN(\beta_j|0, \tau_m^2)\right]^{w_{jm}}.
\end{align}

Our variational approach will be to maximize over
$(f, \bs{\pi}, \phi, \xi)$ the following lower-bound of the
log-marginal likelihood
\begin{align}
\label{eq:elbo}\log p(\hat{\vbeta}|\bs{\pi}, \phi, \xi) \geq \int f(\bs{\beta},\bs{W},\bs{v}) \log\left(\frac{p(\hat{\vbeta}, \bs{\beta}, \bs{W}, \bs{v}|\bs{\pi}, \phi, \xi)}{f(\bs{\beta},\bs{W},\bs{v})}\right)\dif \bs{\beta}\dif\bs{W}\dif\bs{v},
\end{align}
where $f$ an element of some constrained class of densities and
\begin{align}
p(\hat{\vbeta}, \bs{\beta}, \bs{W}, \bs{v}|\bs{\pi}, \phi, \xi) = p(\hat{\vbeta}| \bs{\beta}, \bs{v}, \phi, \xi) p(\bs{\beta}, \bs{W}|\bs{\pi}) p(\bs{v}).
\end{align}
We perform mean-field variational inference and constrain $f$ to be factorized by
\begin{align}
  f(\bs{\beta},\bs{W},\bs{v}) = f(\bs{v})\prod_{j = 1}^p f(\beta_j, \bs{w}_j).
\end{align}
This is the only assumption that we place on the form of the
variational densities. Here, we are indexing the variational densities
by their arguments.  After maximizing \eqref{eq:elbo} over
$(f, \bs{\pi}, \phi, \xi)$, we use the $f(\beta_j, \bs{w}_j)$'s to
provide posterior summaries for the $\beta_j$'s.

The variational updates for all parameters involved are presented in
Algorithm \ref{algorithm:vem}. As the derivations are standard and
tedious, we place the details in Appendix
\ref{section:supp.vem}, though we make a few comments here. First, the
variational density of $\bs{v}$ is a multivariate normal which we
parameterize with mean $\bs{\mu}_{\bs{v}}$ and covariance
$\bs{\Sigma}_{\bs{v}}$. The variational densities of the $\beta_j$'s
turn out to be mixtures of Gaussians which we parameterize with mixing
means $\mu_{jm}$, mixing variances $\sigma_{jm}$, and mixing
proportions $\gamma_{jm}$.  Importantly, if the prior on the
$\beta_j$'s contains a $\tau_m$ that is $0$, representing a pointmass
at 0, then the variational densities of the $\beta_j$'s also must have
a pointmass at 0. This allows us to return local false discovery
rates. The $\lambda_m$'s in Algorithm \ref{algorithm:vem} are the same
penalties as in Section \ref{section:bells.whistles}. Finally, we do
not need to initialize all parameters. It turns out that it suffices
to initialize the variational means of the $\beta_j$'s, the mean of
$\bs{v}$, the prior mixing proportions $\bs{\pi}$, the ``$g$''
hyperparameter $\phi$, and the variance scaling parameter $\xi$. We
initialize the means of the $\beta_j$'s, denoted
$\bs{\mu}_{\bs{\beta}}$, by the posterior means from fitting \ash{} to
$(\hat{\vbeta}, \bs{S})$ assuming no confounding, and we initialize
$\bs{\mu}_{\bs{v}}$ by regressing the resulting residuals on
$\bs{A}$. It intuitively makes sense to initialize $\xi$ at 1 as this
simply indicates that one has adequate variance estimates $\bs{S}$
obtained during the FA step. The choice of initialization of $\phi$ is
not so clear, but we choose a default of 1. Finally, we use the same
initialization of the $\pi_m$'s as ASH.

\begin{algorithm}
\begin{algorithmic}
\STATE Initialize parameters:
\begin{description}[noitemsep, nolistsep]
\item Initialize $\bs{\mu}_{\bs{\beta}}$ by the posterior means from
  fitting \ash{} on $(\hat{\vbeta}, \bs{S})$.
\item Initialize $\bs{\mu}_{\bs{v}} = (\bs{A}^{\intercal}\bs{S}^{-1}\bs{A})^{-1}\bs{A}^{\intercal}\bs{S}^{-1}(\hat{\vbeta} - \bs{\mu}_{\bs{\beta}})$.
\item Initialize $\xi = 1$.
\item Initialize $\phi = 1$.
\end{description}
\REPEAT
\STATE Set $\bs{r} = \hat{\vbeta} - \phi\bs{A}\bs{\mu}_{\bs{v}}$.
\FOR{$j = 1,\ldots,p$}
\FOR{$m = 0, \ldots, M$}
\STATE Set $\sigma_{jm}^2 = \left(\frac{1}{\tau_m^2} + \frac{1}{\xi s_{jj}^2}\right)^{-1}$.
\STATE Set $\mu_{jm} = r_j\sigma_{jm}^2 / (\xi s_{jj}^2)$.
\STATE Set $\gamma_{jm} = \frac{\pi_m N(r_j | 0, \xi s_{jj}^2 + \tau_m^2)}{\sum_{m = 0}^M\pi_m N(r_j | 0, \xi s_{jj}^2 + \tau_m^2)}$.
\ENDFOR
\STATE Set $\bs{\mu}_{\bs{\beta}j} = \sum_{m = 0}^M \gamma_{jm}\mu_{jm}$.
\ENDFOR
\FOR{$m = 0,\ldots, M$}
\STATE Set $\pi_m = \frac{\sum_{j = 1}^p \gamma_{jm} + \lambda_{m} - 1}{\sum_{m = 0}^M\sum_{j = 1}^p \gamma_{jm} + \sum_{m = 0}^M\lambda_{m} - (M + 1)}$.
\ENDFOR
\STATE Set $\bs{\Sigma}_{\bs{v}} = \left(\frac{\phi^2}{\xi}\bs{A}^{\intercal}\bs{S}^{-1}\bs{A} + \bs{I}_{q}\right)^{-1}$.
\STATE Set $\bs{\mu}_{\bs{v}} = \frac{\phi}{\xi} \bs{\Sigma}_{\bs{v}}\bs{A}^{\intercal}\bs{S}^{-1}(\hat{\vbeta} - \bs{\mu}_{\bs{\beta}})$.
\STATE Set $\phi = \frac{\bs{\mu}_{\bs{v}}^{\intercal}\bs{A}^{\intercal}\bs{S}^{-1}(\hat{\vbeta} - \bs{\mu}_{\bs{\beta}})}{\bs{\mu}_{\bs{v}}^{\intercal}\bs{A}^{\intercal}\bs{S}^{-1}\bs{A}\bs{\mu}_{\bs{v}} + \tr(\bs{A}^{\intercal}\bs{S}^{-1}\bs{A}\bs{\Sigma}_{\bs{v}})}$.
\STATE Set
\begin{align}
\begin{split}
\xi =& \frac{1}{p}\Bigg\{\hat{\vbeta}^{\intercal}\bs{S}^{-1}\hat{\vbeta} + \sum_{j = 1}^P\frac{1}{s_{jj}^2}\sum_{m = 0}^M \gamma_{jm}(\mu_{jm}^2 + \sigma_{jm}^2) + \phi^2 \tr\left(\bs{A}^{\intercal}\bs{S}^{-1}\bs{A}(\bs{\mu}_{\bs{v}}\bs{\mu}_{\bs{v}}^{\intercal} + \bs{\Sigma}_{\bs{v}})\right)\\
&- 2 \hat{\vbeta}^{\intercal}\bs{S}^{-1}\bs{\mu}_{\bs{\beta}} - 2\phi\hat{\vbeta}^{\intercal}\bs{S}^{-1}\bs{A}\bs{\mu}_{\bs{v}} + 2\phi \bs{\mu}_{\bs{\beta}}^{\intercal}\bs{S}^{-1}\bs{A}\bs{\mu}_{\bs{v}}\Bigg\}.
\end{split}
\end{align}
\STATE Calculate the penalized ELBO \eqref{eq:elbo.cond}.
\UNTIL{Convergence}
\end{algorithmic}
\caption{Variational Expectation Maximization algorithm to fit BACKWASH.}
\label{algorithm:vem}
\end{algorithm}

\subsection{Simulation details}
\label{section:signal.details}

We describe here how we simulated the data in Section
\ref{section:simulations}. The procedure is the same as in
\citet{gerard2017unifying}.

First, we took the top $p$ expressed genes from the GTEx RNA-seq data
\citep{gtex2015} and randomly sampled $n$ individuals, yielding an
$n \times p$ count matrix $\bs{Z}$. We then randomly assigned $n/2$
samples to one group and other $n/2$ samples to a second group.  At
this point, all gene expression levels are theoretically unassociated
with the group label as group assignment was done independently of any
gene expression. We used this as one scenario in our simulations
(where $\pi_0 = 1$)

We then added signal to a proportion $(1-\pi_0)$ of genes, randomly
chosen from the set of genes represented in the null data, as
follows. First, we sampled the effect sizes from a $N(0, 0.8^2)$, the
variance being chosen as to make the AUC of all methods neither too
close to 1 nor too close to 0.5. For $j_{\ell} \in \Omega$, the set of
non-null genes, let
\begin{align}
a_{j_1},\ldots,a_{j_{(1-\pi_0) p}} \overset{iid}{\sim} N(0, 0.8^2),
\end{align}
be the effect sizes. For each $j_{\ell}\in\Omega$, we then drew new
counts $w_{ij_{\ell}}$ from $z_{ij_{\ell}}$ by
\begin{align}
w_{ij_{\ell}}|z_{ij_{\ell}} \sim
\begin{cases}
\text{Binomial}(z_{ij_{\ell}}, 2^{a_{j_{\ell}}x_{i2}}) &\text{if } a_{j_{\ell}} < 0 \text{ and } j_{\ell} \in \Omega,\\
\text{Binomial}(z_{ij_{\ell}}, 2^{-a_{j_{\ell}}(1 - x_{i2})}) &\text{if } a_{j_{\ell}} > 0  \text{ and } j_{\ell} \in \Omega\\
\delta(z_{ij_{\ell}}) &\text{if } j_{\ell} \notin \Omega,
\end{cases}
\end{align}
Here, $\delta(a)$ is notation for a point-mass at $a$. We then used
$\bs{W}$ as our new response matrix of counts. To obtain the $\bs{Y}$
in \eqref{eq:full.model}, we simply took a $\log_2$ transformation of
the elements of $\bs{W}$.

The intuition behind this approach is that if the original counts
$z_{ij}$ are Poisson distributed, then the new counts $w_{ij}$ are
also Poisson distributed with $a_j$ being the approximate
$\log_2$-effect between groups. That is, if
$z_{ij} \sim Poisson(\lambda_j)$, then
\begin{align}
[w_{ij} | a_{j}, a_{j} < 0, j \in \Omega] &\sim \text{Poisson}(2^{a_{j}x_{i2}}\lambda_j)\\
[w_{ij} | a_{j}, a_{j} > 0, j \in \Omega] &\sim \text{Poisson}(2^{-a_{j}(1 - x_{i2})}\lambda_j).
\end{align}
Hence,
\begin{align}
E[\log_2(w_{ij}) - \log_2(w_{kj}) | a_{j},\ a_{j} < 0,\ j \in \Omega] &\approx a_{j}x_{i2} - a_{j}x_{k2} = a_j(x_{i2} - x_{k2}), \text{ and}\\
E[\log_2(w_{ij}) - \log_2(w_{kj}) | a_{j},\ a_{j} > 0,\ j \in \Omega] &\approx -a_{j}(1 - x_{i2}) + a_{j}(1 - x_{k2}) = a_j(x_{i2} - x_{k2}).
\end{align}

See also \citet{kvam2012comparison, reeb2013evaluating,
  soneson2013comparison, vandewiel2014shrinkbayes, rocke2015excess}
for similar simulation settings.

\subsection{Analysis using the control genes of \citet{lin2017housekeeping}}
\label{section:lin}

We repeated the analysis of the GTEx data in Section
\ref{section:real.data} using the list of control genes collated by
\citet{lin2017housekeeping}. This list was created using single cell
sequencing data and contains only moderate overlap with the list
developed by \citet{eisenberg2013human}. We observe:
\begin{enumerate}
\item The lfdr estimates for the control gene methods are mostly
  similar when using the two different lists. Compare Figures
  \ref{figure:rank.lfdr} and \ref{figure:rank.lfdr.lin}. Also compare
  Figures \ref{figure:lfdr.full} and \ref{figure:lfdr.full.lin}.
\item The estimates of the proportion of genes that are null are also
  mostly similar when using the two lists. Compare Tables
  \ref{table:pi0hat} and \ref{table:pi0hat.lin}.
\item RUV2 methods improved slightly in the positive control analysis
  when using the list from \citet{lin2017housekeeping}. Compare
  Figures \ref{figure:prop.max} and
  \ref{figure:prop.max.lin}. However, again, most of the methods
  performed similarly in ranking the most significant genes.
\end{enumerate}

The comparable performance of control gene methods when using the
lists of \citet{lin2017housekeeping} and \citet{eisenberg2013human}
does not indicate that these lists are of comparable quality. Recall
that MOUTHWASH and BACKWASH both indicate that the vast majority of
genes are null. Thus, it might be that many lists of ``control genes''
would give similar performance, because the vast majority of these
``control genes'' would indeed by null.

\subsection{$t$-likelihood Variance Inflated CATE}
\label{section:catencv}

\begin{algorithm}
  \caption{EM Algorithm for fitting a regression with $t$-errors}
  \label{algorithm:tem}
  \begin{algorithmic}[1]
    \STATE E-step:~Set
    \begin{align}
      w_{j} = \frac{\nu_j + 1}{(\hat{\beta}_{\mathcal{C}j} - \hat{\bs{\alpha}}_{\mathcal{C}j}^{\intercal}\bs{z}_{(old)})^2 / (\xi_{(old)} s_{\mathcal{C}j}^2) + \nu_j}
    \end{align}
    \STATE M-step:~Let $\bs{W} := \diag(w_{1},\ldots,w_{m})$. Set
    \begin{align}
      \bs{z}_{(new)} &= (\hat{\bs{\alpha}}_{\mathcal{C}}
      \bs{W}\bs{S}_{\mathcal{C}}^{-1}
      \hat{\bs{\alpha}}_{\mathcal{C}}^{\intercal})^{-1}\hat{\bs{\alpha}}_{\mathcal{C}}\bs{W}\bs{S}_{\mathcal{C}}^{-1}\hat{\vbeta}_{\mathcal{C}} \\
      \xi_{(new)} &= \frac{1}{m}\sum_{j = 1}^m \frac{w_{j}}{\sigma_j^2}(\hat{\beta}_{j} - \hat{\bs{\alpha}}_{\mathcal{C}j}^{\intercal}\bs{z}_{(new)})^2
    \end{align}
  \end{algorithmic}
\end{algorithm}

To improve robustness to modeling assumptions, we explored modifying
CATE to use a $t$-likelihood in its second step. This is akin to the
ideas presented in Section \ref{section:bells.whistles}. We replace
\eqref{eq:reduced.model.control} with
\begin{align}
  \label{eq:control.model.var.t}
  [\hat{\beta}_{\mathcal{C}j}|\hat{\bs{\alpha}}_{\mathcal{C}j}^{\intercal}, \bs{z}, \xi, s_{\mathcal{C}j}^2]  \overset{ind}{\sim} t_{\nu_j}(\hat{\bs{\alpha}}_{\mathcal{C}j}^{\intercal}\bs{z}, \xi s_{\mathcal{C}j}^2),
\end{align}
where $t_{\nu_j}(\cdot|a, b^2)$ is as defined in \eqref{eq:t.density}
and $\hat{\bs{\alpha}}_{\mathcal{C}j}$ is the $j$th column of
$\hat{\bs{\alpha}}_{\mathcal{C}}$. The degrees of freedom
($\nu_j$'s) are assumed known. CATE uses
\eqref{eq:reduced.model.control} to estimate $\bs{z}$ by maximum
likelihood. Hence, we use \eqref{eq:control.model.var.t} to estimate
$\bs{z}$ and $\xi$ by maximum likelihood. To do so, we apply an
expectation-maximization (EM) algorithm that is similar to that
discussed in Appendix A.2 of \citet{lange1989robust}. The model
\eqref{eq:control.model.var.t} can be represented by including a
latent variable $\tau_{j}$ for each observation
\begin{align}
  \hat{\beta}_{\mathcal{C}j}|\tau_{j} \sim N(\hat{\bs{\alpha}}_{\mathcal{C}j}^{\intercal}\bs{z}, \tau_j\xi s_{\mathcal{C}j}^2), \ \tau_{j} \sim \text{ Inverse-Gamma}(\nu_j/2, \nu_j/2), \label{eq:latent.t}
\end{align}
Using \eqref{eq:latent.t}, an EM algorithm to fit this model is easily
obtained. One step of this algorithm is presented in Algorithm
\ref{algorithm:tem}. Repeated applications of the step in Algorithm
\ref{algorithm:tem} is guaranteed to increase the likelihood at each
iteration, converging to a local maximum.

\subsection{Derivation of VEM Algorithm}
\label{section:supp.vem}
Here, we derive the updates for the variational EM algorithm presented
in Section \ref{section:vem.backwash}. We begin by writing out all
densities involved:
\begin{align}
&p(\hat{\vbeta}, \bs{\beta}, \bs{W}, \bs{v}|\bs{\pi}, \xi, \phi) = p(\hat{\vbeta}|\bs{\beta}, \bs{v}, \xi, \phi) p(\bs{\beta}, \bs{W}|\bs{\pi}) p(\bs{v}),\\
&p(\hat{\vbeta}|\bs{\beta}, \bs{v}, \xi, \phi) = (2\pi)^{-p/2} \xi^{-p/2} \det(\bs{S})^{-1/2} \exp\left(-\frac{1}{2\xi}(\hat{\vbeta} - \bs{\beta} - \phi\bs{A}\bs{v})^{\intercal}\bs{S}^{-1}(\hat{\vbeta} - \bs{\beta} - \phi\bs{A}\bs{v})\right),\\
&p(\bs{\beta}, \bs{W}|\bs{\pi}) = \prod_{j = 1}^p\prod_{m = 0}^M \left\{\pi_m (2\pi\tau_m^2)^{-1/2}\exp\left(-\frac{1}{2\tau_m^2}\beta_j^2\right)\right\}^{w_{jm}},\\
&p(\bs{v}) = (2\pi)^{-q/2}\exp\left(-\frac{1}{2}\bs{v}^{\intercal}\bs{v}\right),\\
&h(\bs{\pi}) = \prod_{m = 0}^M\pi_m^{\lambda_m - 1}.
\end{align}

\paragraph{Update of $f(\beta_j, \bs{w}_j)$:} By a general result in
mean-field variational inference \citep[see][for
example]{bishop2006pattern} we update $f(\beta_j, \bs{w}_j)$ by
\begin{align}
\label{eq:first.betaw}
\log f(\beta_j, \bs{w}_j) \propto E_{-(\beta_j, \bs{w}_j)}\left[\log p(\hat{\vbeta}, \bs{\beta}, \bs{W}, \bs{v}|\bs{\pi}, \xi, \phi)\right],
\end{align}
where ``$\propto$'' here denotes that the relationship holds up to an
additive constant that does not depend on $(\beta_j, \bs{w}_j)$, and
$E_{-(\beta_j, \bs{w}_j)}[\cdot]$ denotes that we take the expectation
with respect to all variational densities except that of
$(\beta_j, \bs{w}_j)$. Let $\bs{r} := \hat{\vbeta} -
\phi\bs{A}E[\bs{v}]$. Then we have
\begin{align}
\eqref{eq:first.betaw} &\propto E_{-(\beta_j, \bs{w}_j)}\left[\log p(\hat{\vbeta}|\bs{\beta}, \bs{v}, \xi, \phi) + \log p(\beta_j, \bs{w}_j|\bs{\pi})\right]\\
&\propto E_{-(\beta_j, \bs{w}_j)}\left[-\frac{1}{2\xi}(\hat{\vbeta} - \bs{\beta} - \phi\bs{A}\bs{v})^{\intercal}\bs{S}^{-1}(\hat{\vbeta} - \bs{\beta} - \phi\bs{A}\bs{v})\right] + \log p(\beta_j, \bs{w}_j|\bs{\pi})\\
&\propto -\frac{1}{2\xi s_{jj}}\left(\beta_j^2 - 2\beta_jr_j\right) +\log p(\beta_j, \bs{w}_j|\bs{\pi})\\
&\propto \log\left(N(r_j|\beta_j, \xi s_{jj}^2)\right) + \log p(\beta_j, \bs{w}_j|\bs{\pi})\\
&\propto \log\left(N(r_j|\beta_j, \xi s_{jj}^2)\right) + \sum_{m = 0}^M w_{jm}\log\left(\pi_mN(\beta_j|0, \tau_m^2)\right)\\
&\propto \sum_{m = 0}^M w_{jm}\log\left(\pi_mN(r_j|\beta_j, \xi s_{jj}^2)N(\beta_j|0, \tau_m^2)\right)\\
\label{eq:bayes.argument}&\propto \sum_{m = 0}^M w_{jm}\log\left(\pi_mN(r_j|0, \xi s_{jj}^2 + \tau_m^2)N(\beta_j|\mu_{jm}, \sigma_{jm}^2)\right),
\end{align}
where
\begin{align}
\sigma_{jm}^2 &:= \left(\frac{1}{\tau_m^2} + \frac{1}{\xi s_{jj}^2}\right)^{-1}, \text{ and}\\
\mu_{jm} &:= r_j \sigma_{jm}^2 / (\xi s_{jj}^2).
\end{align}
Equation \eqref{eq:bayes.argument} follows by standard Bayesian
conjugacy arguments. Equation \eqref{eq:bayes.argument} is the
log-kernel of a density of a mixture of normals with mixing means
$\mu_{jm}$ for $m = 0,\ldots,M$ and mixing variances $\sigma_{jm}^2$
for $m = 0,\ldots,M$. The mixing weights are proportional to
$\pi_mN(r_j|0, \xi s_{jj}^2 + \tau_m^2)$. Since the mixing weights
must sum to unity we have that they are
\begin{align}
\gamma_{jm} := \frac{\pi_mN(r_j|0, \xi s_{jj}^2 + \tau_m^2)}{\sum_{m = 0}^M \pi_mN(r_j|0, \xi s_{jj}^2 + \tau_m^2)}.
\end{align}

\paragraph{Update $f(\bs{v})$:}
Again, using a standard argument from mean-field variational
inference, we update the variational density of $\bs{v}$ with
\begin{align}
\log f(\bs{v}) &\propto E_{-\bs{v}}\left[\log p(\hat{\vbeta}, \bs{\beta}, \bs{W}, \bs{v}|\bs{\pi}, \xi, \phi)\right]\\
&\propto E_{-\bs{v}}\left[\log p(\hat{\vbeta}|\bs{\beta}, \bs{v}, \xi, \phi) + \log p(\bs{v})\right]\\
&\propto E\left[-\frac{1}{2\xi}(\hat{\vbeta} - \bs{\beta} - \phi\bs{A}\bs{v})^{\intercal}\bs{S}^{-1}(\hat{\vbeta} - \bs{\beta} - \phi\bs{A}\bs{v})\right] - \frac{1}{2}\bs{v}^{\intercal}\bs{v}\\
&\propto -\frac{1}{2}\left(\frac{\phi^2}{\xi}\bs{v}^{\intercal}\bs{A}^{\intercal}\bs{S}^{-1}\bs{A}\bs{v} - 2 \frac{\phi}{\xi}\bs{v}^{\intercal}\bs{A}^{\intercal}\bs{S}(\hat{\vbeta} - E[\bs{\beta}])\right) - \frac{1}{2}\bs{v}^{\intercal}\bs{v}\\
\label{eq:kernel.v}&\propto -\frac{1}{2}\left[\bs{v}^{\intercal}\left(\frac{\phi^2}{\xi}\bs{A}^{\intercal}\bs{S}^{-1}\bs{A} + \bs{I}_{q}\right)\bs{v} - 2 \frac{\phi}{\xi}\bs{v}^{\intercal}\bs{A}^{\intercal}\bs{S}(\hat{\vbeta} - E[\bs{\beta}])\right].
\end{align}
Equation \eqref{eq:kernel.v} is the log-kernel of a multivariate
normal density with covariance matrix $\bs{\Sigma}_{\bs{v}}$ and mean
$\bs{\mu}_{\bs{v}}$, where
\begin{align}
\bs{\Sigma}_{\bs{v}} &:= \left(\frac{\phi^2}{\xi}\bs{A}^{\intercal}\bs{S}^{-1}\bs{A} + \bs{I}_{q}\right)^{-1},\text{ and}\\
\bs{\mu}_{\bs{v}} &:= \frac{\phi}{\xi}\bs{\Sigma}_{\bs{v}}\bs{A}^{\intercal}\bs{S}^{-1}(\hat{\vbeta} - E[\bs{\beta}]).
\end{align}

\paragraph{Update $\phi$:}
We update $\phi$ by finding
\begin{align}
\phi^{(new)} &= \argmax_{\phi} E\left[\log p(\hat{\vbeta}, \bs{\beta}, \bs{W}, \bs{v}|\bs{\pi}, \xi, \phi)\right]\\
&= \argmax_{\phi} E\left[\log p(\hat{\vbeta}|\bs{\beta}, \bs{v}, \xi, \phi)\right]\\
&= \argmax_{\phi} E\left[-\frac{1}{2\xi}(\hat{\vbeta} - \bs{\beta} - \phi\bs{A}\bs{v})^{\intercal}\bs{S}^{-1}(\hat{\vbeta} - \bs{\beta} - \phi\bs{A}\bs{v})\right]\\
&= \argmin_{\phi} \left\{\phi^2\tr\left(\bs{A}^{\intercal}\bs{S}^{-1}\bs{A}E[\bs{v}\bs{v}^{\intercal}]\right) - 2 \phi E[\bs{v}]^{\intercal}\bs{A}^{\intercal}\bs{S}^{-1}(\hat{\vbeta} - E[\bs{\beta}])\right\}\\
&= \frac{E[\bs{v}]^{\intercal}\bs{A}^{\intercal}\bs{S}^{-1}(\hat{\vbeta} - E[\bs{\beta}])}{\tr\left(\bs{A}^{\intercal}\bs{S}^{-1}\bs{A}E[\bs{v}\bs{v}^{\intercal}]\right)}.
\end{align}

\paragraph{Update $\xi$:}
We update $\xi$ by finding
\begin{align}
\xi^{(new)} &= \argmax_{\xi} E\left[\log p(\hat{\vbeta}, \bs{\beta}, \bs{W}, \bs{v}|\bs{\pi}, \xi, \phi)\right]\\
&= \argmax_{\xi} E\left[\log p(\hat{\vbeta}|\bs{\beta}, \bs{v}, \xi, \phi)\right]\\
&= \argmax_{\xi} \left\{-\frac{p}{2}\log(\xi)-\frac{1}{2\xi}E\left[(\hat{\vbeta} - \bs{\beta} - \phi\bs{A}\bs{v})^{\intercal}\bs{S}^{-1}(\hat{\vbeta} - \bs{\beta} - \phi\bs{A}\bs{v})\right]\right\}\\
&= \frac{1}{p}E\left[(\hat{\vbeta} - \bs{\beta} - \phi\bs{A}\bs{v})^{\intercal}\bs{S}^{-1}(\hat{\vbeta} - \bs{\beta} - \phi\bs{A}\bs{v})\right].
\end{align}

\paragraph{Update $\bs{\pi}$:}
Finally, we update $\bs{\pi}$ by
\begin{align}
\bs{\pi}^{(new)} &= \argmax_{\bs{\pi}} E\left[\log p(\hat{\vbeta}, \bs{\beta}, \bs{W}, \bs{v}|\bs{\pi}, \xi, \phi)\right] + \log(h(\bs{\pi}))\\
&= \argmax_{\bs{\pi}} E\left[\log p(\bs{W}|\bs{\pi})\right] + \log(h(\bs{\pi}))\\
&= \argmax_{\bs{\pi}} \sum_{m = 0}^M\left(\sum_{j = 1}^p E[w_{jm}] + \lambda_m - 1\right)\log(\pi_{m}).
\end{align}
Hence, we have
\begin{align}
\pi_{m}^{(new)} = \frac{\sum_{j = 1}^p E[w_{jm}] + \lambda_m - 1}{\sum_{m = 0}^M\left(\sum_{j = 1}^p E[w_{jm}] + \lambda_m - 1\right)}.
\end{align}
All of the expectations in the above updates are tedious to compute
but standard so we omit the details.

For our convergence criterion, we monitor the increase in the
lower-bound of the log-marginal likelihood \eqref{eq:elbo}. It can be
written in closed form as
\begin{align}
\begin{split}
\label{eq:elbo.cond}
\int &f(\bs{\beta},\bs{W},\bs{v}) \log\left(\frac{p(\hat{\vbeta}, \bs{\beta}, \bs{W}, \bs{v}|\bs{\pi}, \phi, \xi)}{f(\bs{\beta},\bs{W},\bs{v})}\right)\dif \bs{\beta}\dif\bs{W}\dif\bs{v}\\
=& -\frac{p}{2}\log(\xi) - \frac{1}{2\xi}\Bigg\{\hat{\vbeta}^{\intercal}\bs{S}^{-1}\hat{\vbeta} + \sum_{j = 1}^P\frac{1}{s_{jj}^2}\sum_{m = 0}^M \gamma_{jm}(\mu_{jm}^2 + \sigma_{jm}^2) + \phi^2 \tr\left(\bs{A}^{\intercal}\bs{S}^{-1}\bs{A}(\bs{\mu}_{\bs{v}}\bs{\mu}_{\bs{v}}^{\intercal} + \bs{\Sigma}_{\bs{v}})\right)\\
&- 2 \hat{\vbeta}^{\intercal}\bs{S}^{-1}\bs{\mu}_{\bs{\beta}} - 2\phi\hat{\vbeta}^{\intercal}\bs{S}^{-1}\bs{A}\bs{\mu}_{\bs{v}} + 2\phi \bs{\mu}_{\bs{\beta}}^{\intercal}\bs{S}^{-1}\bs{A}\bs{\mu}_{\bs{v}}\Bigg\}\\
&+ \sum_{j = 1}^p\left\{\gamma_{j0}\log(\pi_0) + \sum_{m = 1}^M \gamma_{jm}\left(\log(\pi_m) - \frac{1}{2}\log(2\pi) - \frac{1}{2}\log(\tau_m^2) - \frac{1}{2\tau_m^2}(\mu_{jm}^2 + \sigma_{jm}^2)\right)\right\}\\
&- \frac{1}{2}\bs{\mu}_{\bs{v}}^{\intercal}\bs{\mu}_{\bs{v}} - \frac{1}{2}\tr(\bs{\Sigma}_{\bs{v}}) + \sum_{m = 0}^M (\lambda_m - 1)\log(\pi_m)\\
& + \frac{1}{2}\log\det(\bs{\Sigma}_{\bs{v}}) - \sum_{j = 1}^p \left\{\gamma_{j0}\log(\gamma_{j0}) +  \sum_{m = 1}^M \gamma_{jm}\left(\log(\gamma_{jm}) - \frac{1}{2}\log(2\pi) - \frac{1}{2}\log(\sigma_{jm}^2) - \frac{1}{2}\right)\right\}\\
&+ \text{ constant},
\end{split}
\end{align}
where ``constant'' indicates an additive constant that is independent
of all parameters that we are optimizing over.

\subsection{Proof of Lemma \ref{lemma:deconvolution}}
\label{section:lemma.proof}

We make use of the following results from
\citet{lukacs1970characteristic}.
\begin{theorem}[Theorem 2.1.1 from \citet{lukacs1970characteristic}]
\label{theorem:char.prop}
Let $F(x)$ be a distribution function with characteristic function $f(t)$. Then
\begin{enumerate}[noitemsep, nolistsep]
\item $f(0) = 1$,
\item $|f(t)| \leq 1$,
\end{enumerate}
where $|\cdot|$ denotes the modulus.
\end{theorem}

\begin{theorem}[Theorem 4.5.1 from \citet{lukacs1970characteristic}
  due to Aleksandr Yakovlevich Khinchin]
  A distribution function is unimodal with vertex $x = 0$ if, and only
  if, its characteristic function $f(t)$ can be represented as
\begin{align}
\label{eq:uni.char.rep}
f(t) = \frac{1}{t}\int_0^th(u)\dif u\ (-\infty \leq t \leq \infty),
\end{align}
where $h(u)$ is a characteristic function.
\end{theorem}

\begin{proof}[Proof of Lemma \ref{lemma:deconvolution}]
  Let $f(t)$ be the characteristic function of $F$ and let $g(t)$ be
  the characteristic function of $G$. Recall that the characteristic
  function of a $N(0,\xi)$ random variable is
\begin{align}
  \label{eq:normal.char}
  k(t) := e^{-\frac{1}{2}\xi t^2}.
\end{align}
Since $F$ is a convolution of $G$ and a $N(0, \xi)$ distribution
function, we have
\begin{align}
f(t) = e^{-\frac{1}{2}\xi t^2}g(t) \Rightarrow g(t) = e^{\frac{1}{2}\xi t^2}f(t).
\end{align}
Since $f(t)$ is unimodal about $0$, we use representation
\eqref{eq:uni.char.rep} and write
\begin{align}
\label{eq:g.n.h.2}
g(t) = e^{\frac{1}{2}\xi t^2}\frac{1}{t}\int_0^th(u)\dif u,
\end{align}
where $h(t)$ is a characteristic function. Using integration by parts,
we can write \eqref{eq:g.n.h.2} as
\begin{align}
\label{eq:g.n.h.3}
g(t) = \frac{1}{t}\int_{0}^{t}\left[\xi u^2e^{\frac{1}{2}\xi u^2} \frac{1}{u}\int_{0}^{u}h(v)\dif v + e^{\frac{1}{2}\xi u^2} h(u)\right]\dif u.
\end{align}
We now show that the integrand in \eqref{eq:g.n.h.3}
is not a characteristic function for sufficiently large $\xi$. Using
\eqref{eq:uni.char.rep} and \eqref{eq:normal.char}, we can
write the integrand in \eqref{eq:g.n.h.3} as
\begin{align}
\label{eq:char.integrand}
\xi u^2 k(u)f(u) + k(u) h(u).
\end{align}
Since
\begin{align}
\label{eq:modulus.integrand}
|\xi u^2 k(u)f(u) + k(u) h(u)| = \xi |u^2 k(u)f(u) + k(u) h(u) / \xi|,
\end{align}
it is now clear that for any fixed non-zero $u$, the limit of
\eqref{eq:modulus.integrand} as $\xi \rightarrow \infty$ is
$\infty$. Thus, for any fixed non-zero $u$, we can make $\xi$ large
enough so that the modulus of \eqref{eq:char.integrand} is larger than
$1$, violating property $2$ of Theorem \ref{theorem:char.prop}. Thus,
for large enough $\xi$, \eqref{eq:char.integrand} is not a
characteristic function.

It remains to note that the integrand in \eqref{eq:g.n.h.3} is
unique up to a set of Lebesgue measure 0. That is, if
\begin{align}
g(t) = \frac{1}{t} \int_{0}^{t} q(u) \dif u\ (-\infty \leq t \leq \infty),
\end{align}
then
\begin{align}
q(u) = \xi u^2e^{\frac{1}{2}\xi u^2} \frac{1}{u}\int_{0}^{u}h(v)\dif v + e^{\frac{1}{2}\xi u^2} h(u),
\end{align}
except on a set of Lebesgue measure zero. Thus,
$q(u) \underset{\xi \rightarrow \infty}{\longrightarrow} \infty$
almost everywhere, and so there is no choice of $q(u)$ that is a
characteristic function for all $\xi$. Hence, for large enough $\xi$,
$G$ is not unimodal.
\end{proof}

\clearpage
\subsection{Supplementary Figures}
\newcounter{sfigure}
\renewcommand{\thefigure}{S\arabic{sfigure}}

\newcounter{stable}
\renewcommand{\thetable}{S\arabic{stable}}

\stepcounter{sfigure}
\begin{figure}[!htb]
\begin{center}
\includegraphics[scale = 0.9]{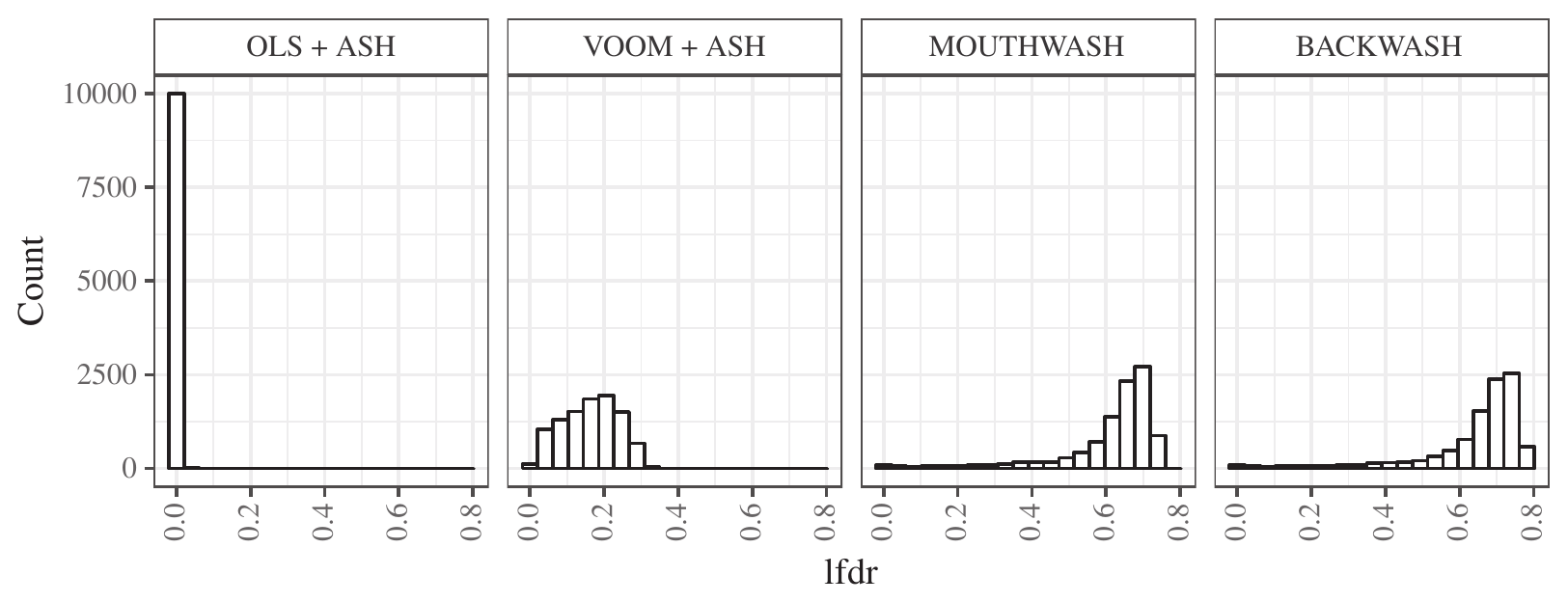}
\end{center}
\caption{Histograms of lfdr for four methods applied to a single
  simulated null dataset.  From left to right:~OLS followed by ASH; a
  voom transformation followed by limma and hierarchical shrinkage of
  variances \protect\citep{law2014voom} followed by ASH; MOUTHWASH;
  and BACKWASH.}
\label{figure:lfdr}
\end{figure}

\stepcounter{sfigure}
\begin{figure}[!htb]
\begin{center}
\includegraphics[scale = 0.9]{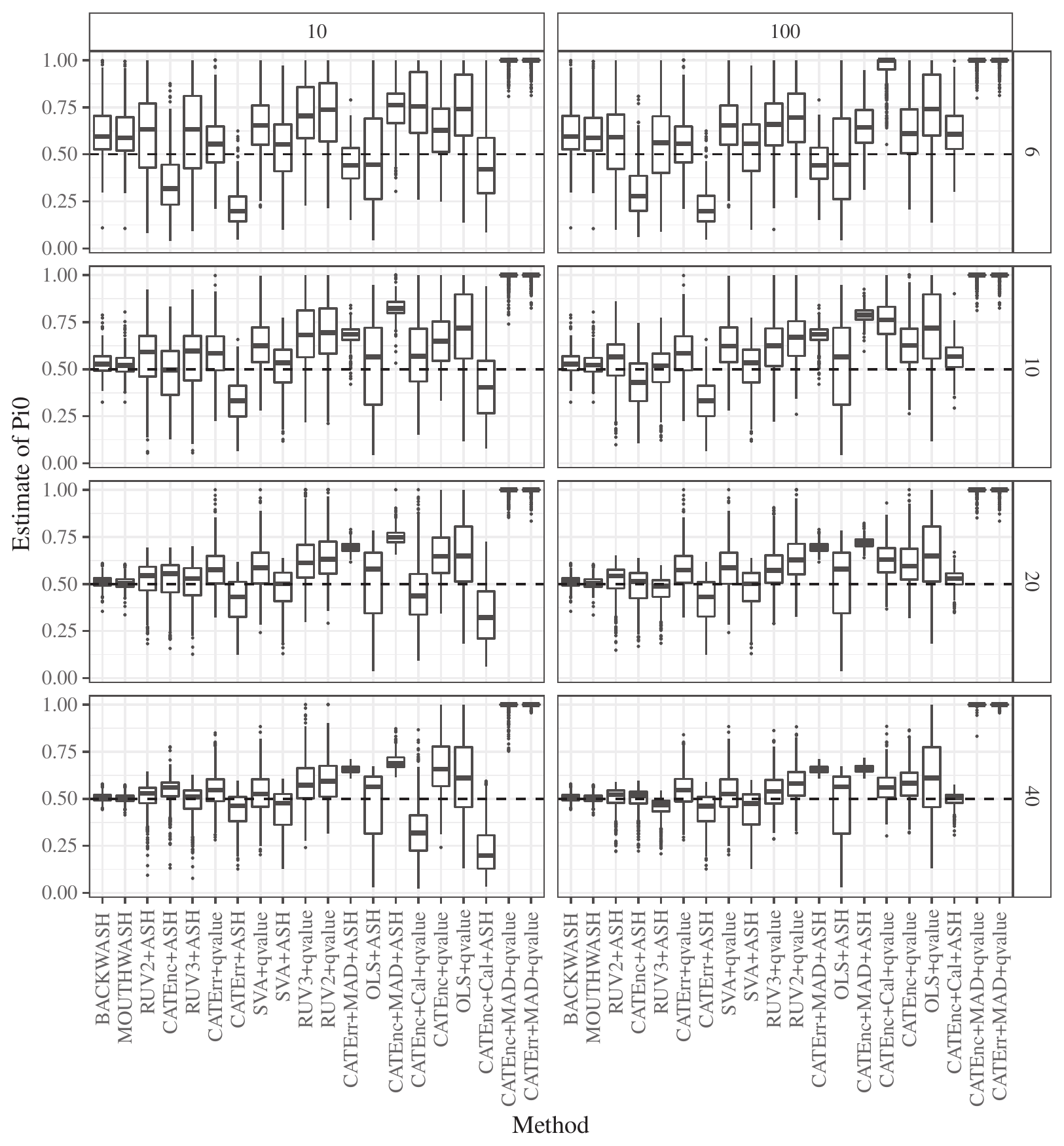}
\end{center}
\caption{Boxplots of estimates of $\pi_0$ for all the methods when
  $\pi_0 = 0.5$. The rows are the sample sizes, the columns are the
  number of control genes used (for methods that use control
  genes). The methods are ordered by the their mean squared error in
  the case when there are 10 control genes and the sample size is
  40. The dashed horizontal line has a $y$-intercept at 0.5}
\label{figure:pi0.5}
\end{figure}

\stepcounter{sfigure}
\begin{figure}[!htb]
\begin{center}
\includegraphics[scale = 0.9]{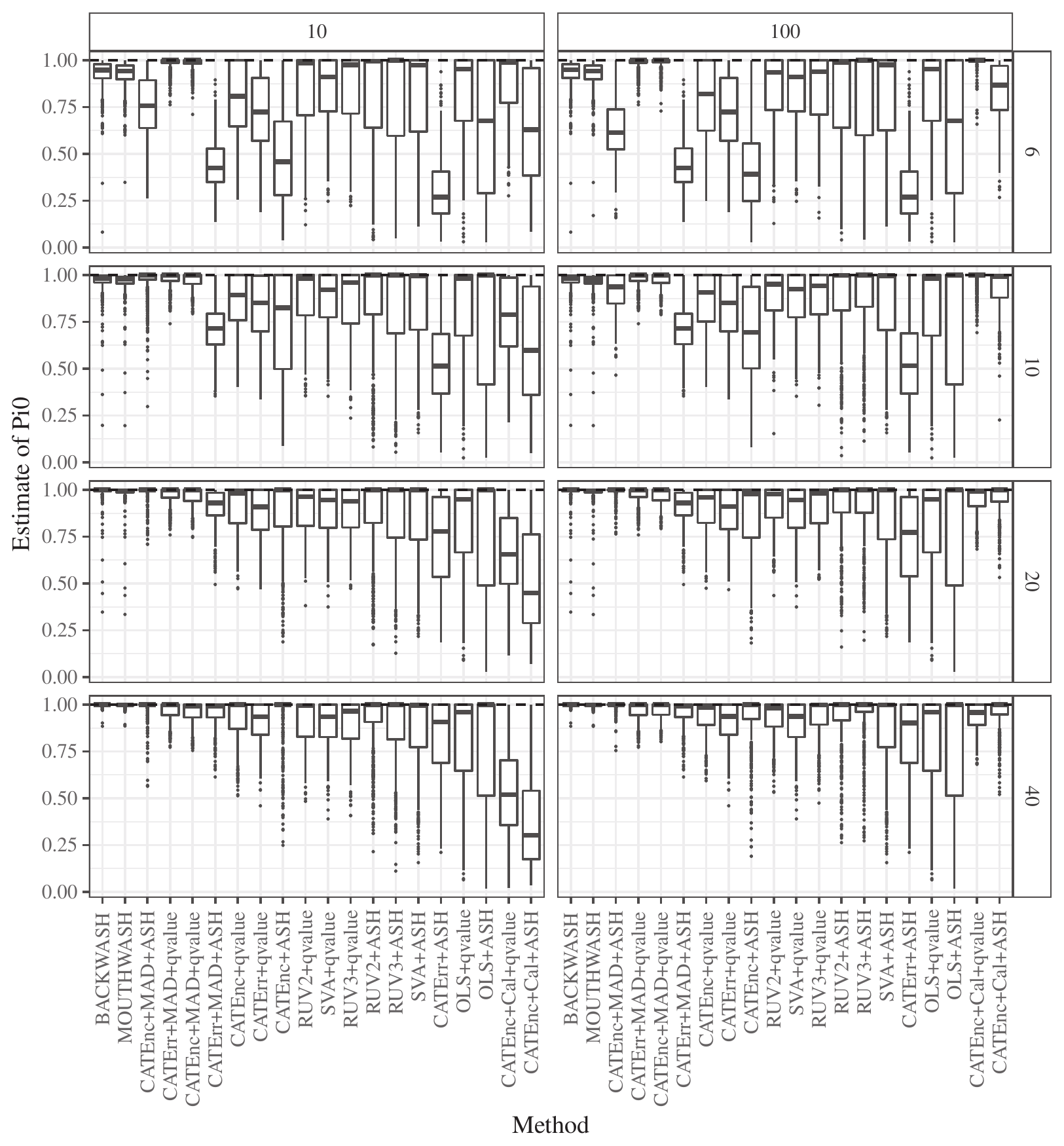}
\end{center}
\caption{Boxplots of estimates of $\pi_0$ for all the methods when
  $\pi_0 = 1$. The rows are the sample sizes, the columns are the
  number of control genes used (for methods that use control
  genes). The methods are ordered by the their mean squared error in
  the case when there are 10 control genes and the sample size is
  40. The dashed horizontal line has a $y$-intercept at 1}
\label{figure:pi0.1}
\end{figure}

\stepcounter{stable}
\begin{table}[ht]
  \centering
  \caption{Computation time, in seconds, of the methods fit in Section
    \ref{section:simulations} when $n = 100$ and $p = 10\text{,}000$. The
    ``Time'' column contains the 0.5, 0.025, and 0.975 quantiles of
    computation time over 100 replicates.}
  \label{tab:comp.time}
\begin{tabular}{ll}
  \hline
  Method    & Time (sec) \\
  \hline
  OLS       & 0.03 (0.03, 0.04) \\
  RUV2      & 0.08 (0.08, 0.11) \\
  CATErr    & 0.18 (0.17, 0.24) \\
  CATEnc    & 0.64 (0.62, 0.88) \\
  RUV3      & 1.13 (1.1, 1.35) \\
  SVA       & 2.09 (2.07, 2.76) \\
  MOUTHWASH & 139.85 (134.57, 150.1) \\
   \hline
\end{tabular}
\end{table}

\stepcounter{sfigure}
\begin{figure}[!htb]
\begin{center}
\includegraphics[scale = 0.9]{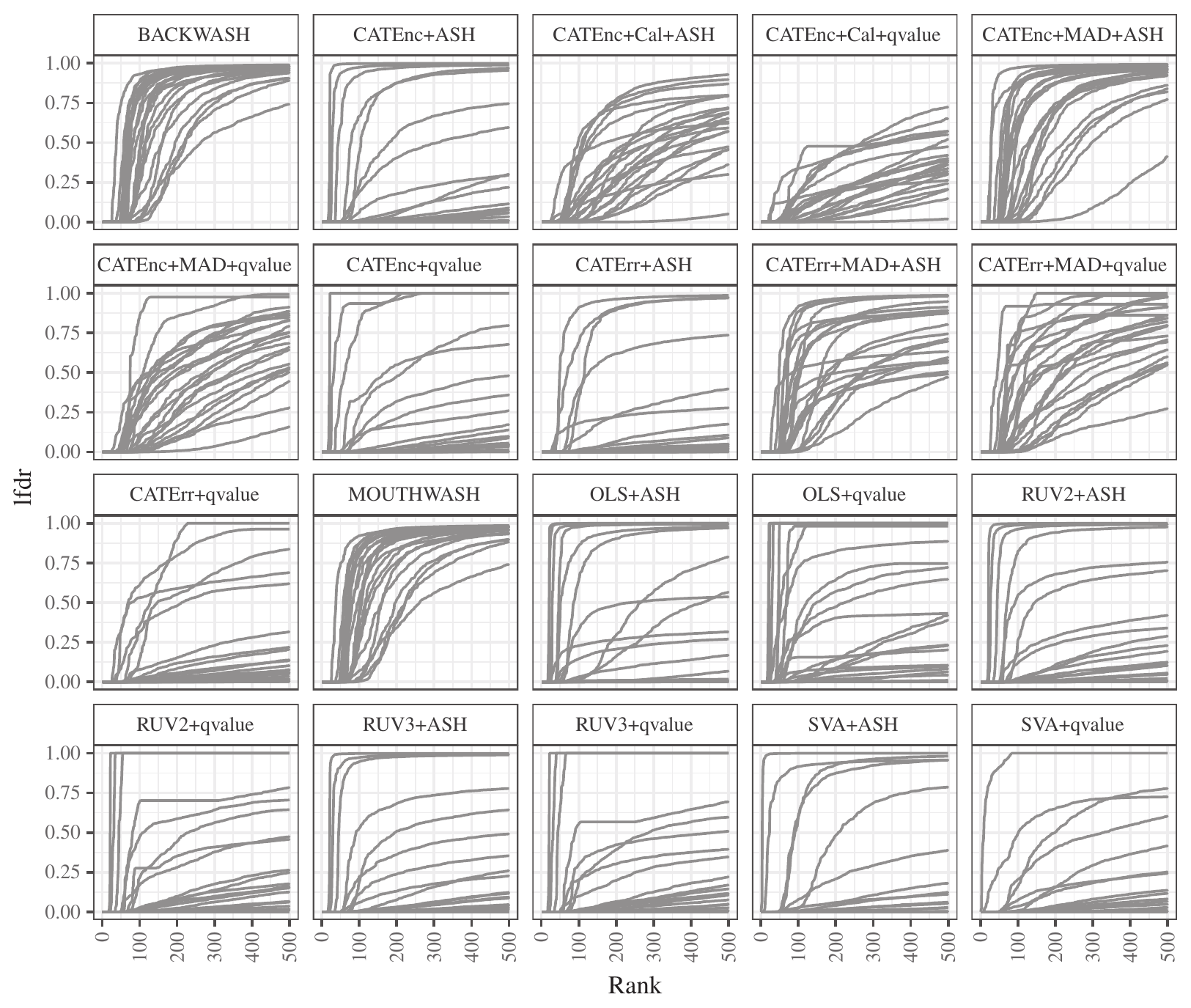}
\end{center}
\caption{Smallest 500 lfdr's versus rank for each method in each
  tissue from the GTEx data. Each facet is a different method and each
  line is a different tissue.}
\label{figure:lfdr.full}
\end{figure}

\stepcounter{sfigure}
\begin{figure}[!htb]
\begin{center}
\includegraphics[scale=0.9]{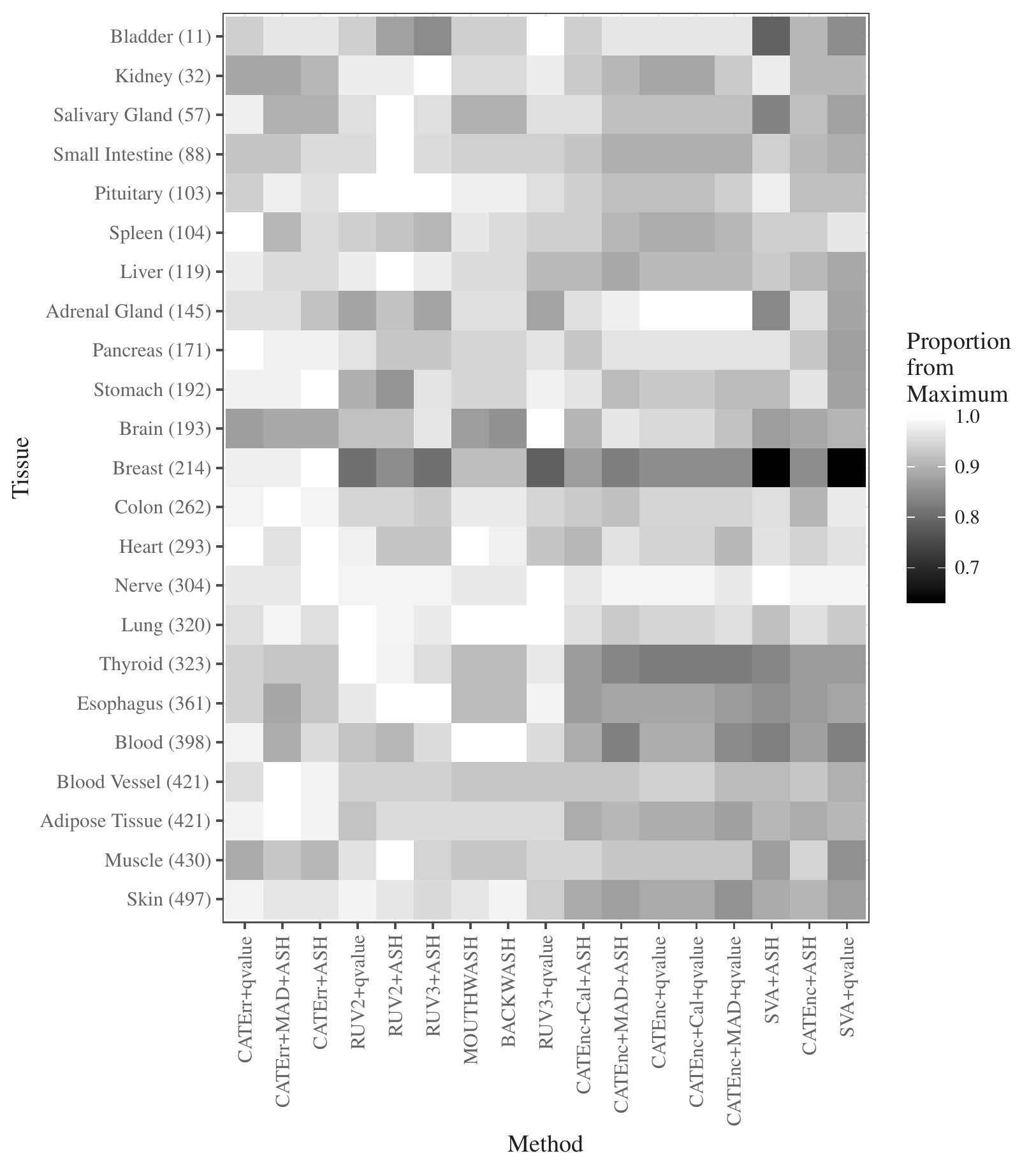}
\caption{This is a repeat of Figure \ref{figure:prop.max} except the
  control gene methods use the list from
  \protect\citet{lin2017housekeeping}. See Figure
  \ref{figure:prop.max} for a description.}
\label{figure:prop.max.lin}
\end{center}
\end{figure}

\stepcounter{stable}
\begin{table}[ht]
  \centering
  \caption{Median estimate of $\pi_0$ for each method across tissues
    when testing for differences between sexes. This is the same table
    as Table \ref{table:pi0hat} except the control gene methods used
    the list from \protect\citet{lin2017housekeeping}.}
  \label{table:pi0hat.lin}
\begin{tabular}{lr}
  \hline
Method & $\hat{\pi}_0$ \\
  \hline
SVA+ASH & 0.28 \\
  CATErr+ASH & 0.33 \\
  RUV3+ASH & 0.39 \\
  OLS+ASH & 0.40 \\
  RUV2+ASH & 0.40 \\
  CATEnc+ASH & 0.49 \\
  SVA+qvalue & 0.70 \\
  CATEnc+Cal+ASH & 0.71 \\
  CATErr+qvalue & 0.76 \\
  CATEnc+qvalue & 0.77 \\
  RUV2+qvalue & 0.78 \\
  RUV3+qvalue & 0.78 \\
  OLS+qvalue & 0.80 \\
  CATEnc+Cal+qvalue & 0.87 \\
  CATErr+MAD+ASH & 0.91 \\
  MOUTHWASH & 0.99 \\
  CATEnc+MAD+ASH & 0.99 \\
  BACKWASH & 0.99 \\
  CATEnc+MAD+qvalue & 1.00 \\
  CATErr+MAD+qvalue & 1.00 \\
   \hline
\end{tabular}
\end{table}

\stepcounter{sfigure}
\begin{figure}[!htb]
\begin{center}
\includegraphics[scale = 0.9]{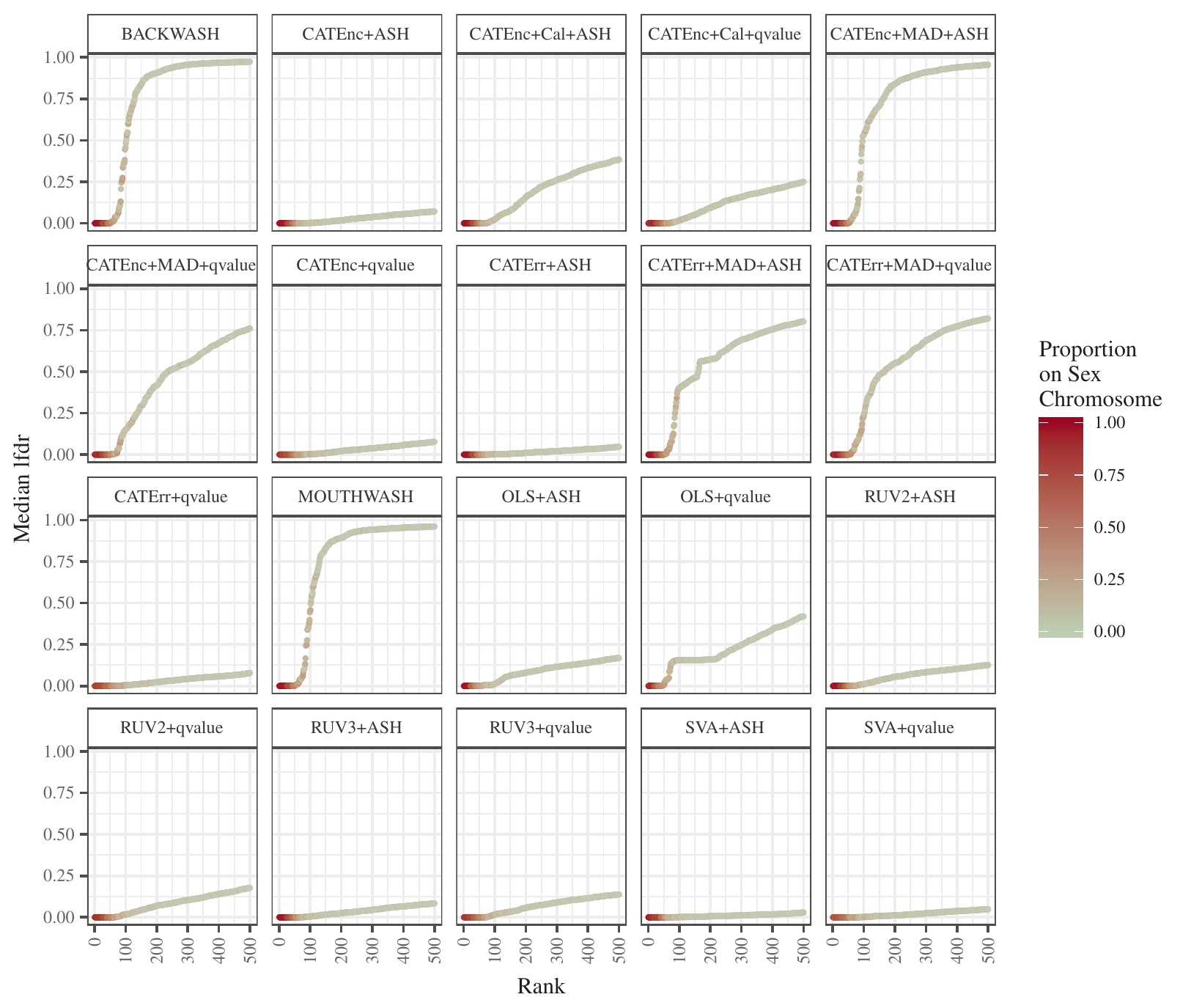}
\caption{This is a repeat of Figure \ref{figure:rank.lfdr} except the
  control gene methods use the list from
  \protect\citet{lin2017housekeeping}. See Figure
  \ref{figure:rank.lfdr} for a description.}
\label{figure:rank.lfdr.lin}
\end{center}
\end{figure}

\stepcounter{sfigure}
\begin{figure}[!htb]
\begin{center}
\includegraphics[scale = 0.9]{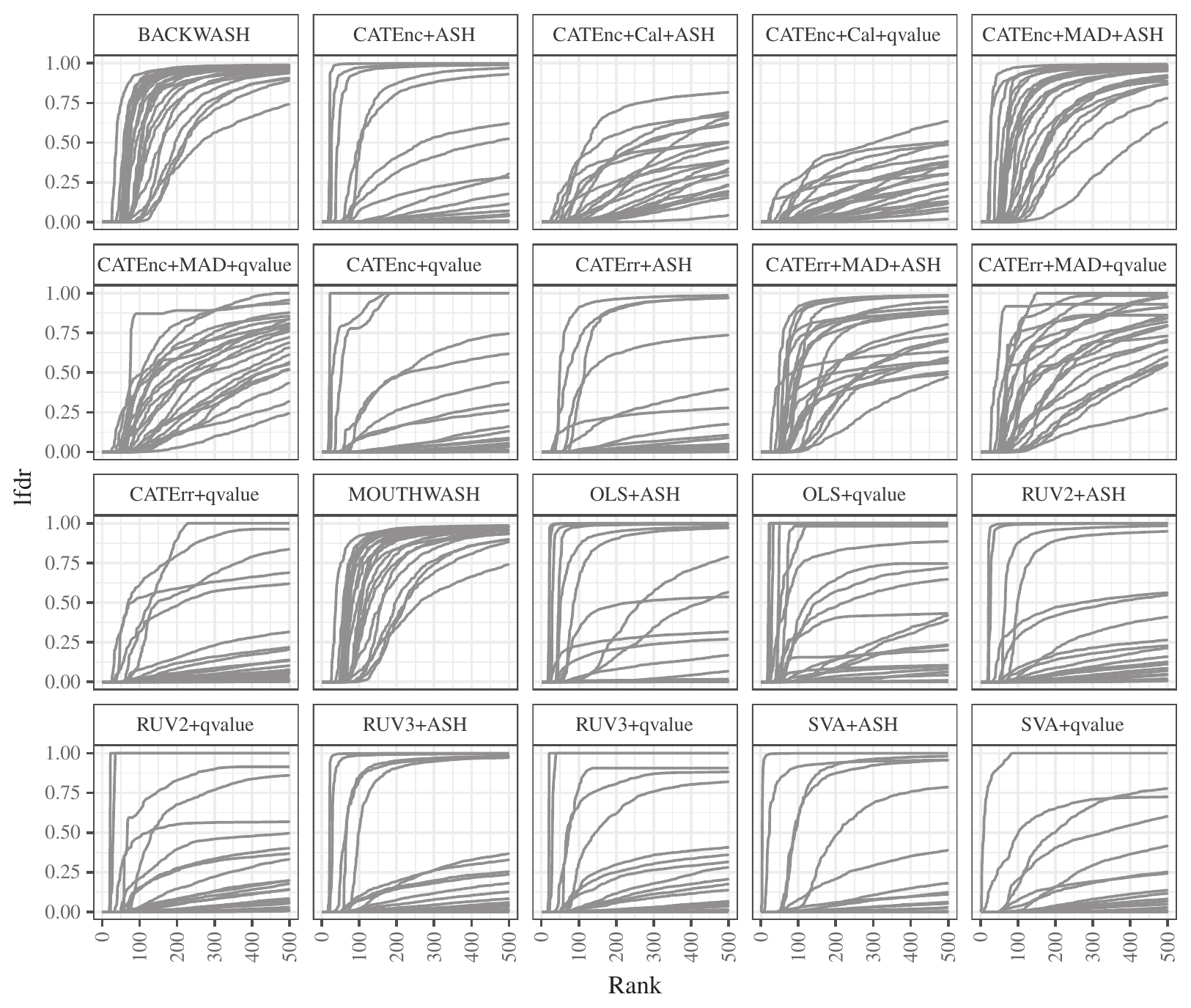}
\end{center}
\caption{This is a repeat of Figure \ref{figure:lfdr.full} except the
  control gene methods use the list from
  \protect\citet{lin2017housekeeping}. See Figure
  \ref{figure:lfdr.full} for a description.}
\label{figure:lfdr.full.lin}
\end{figure}

\clearpage

\bibliography{vicar_bib}

\end{document}